\documentclass[a4paper, intlimits, 10pt]{amsart}

\usepackage[T1]{fontenc}
\usepackage{float}

\usepackage{times}
\usepackage[english]{babel}
\usepackage[T1]{fontenc}
\usepackage{enumitem}

\usepackage{graphicx}

\usepackage{subcaption}
\usepackage{booktabs}
\usepackage{tabularx}
\usepackage{color}
\usepackage{eurosym}

\usepackage[numbers,sort&compress]{natbib}

\usepackage{amsthm, amssymb, amsfonts}

\usepackage{accents}
\newcommand{\ubar}[1]{\underaccent{\bar}{#1}}

\newtheorem{theorem}{Theorem}[section]

\newtheorem{proposition}[theorem]{Proposition}

\newtheorem{remark}[theorem]{Remark}
\numberwithin{equation}{section}

\begin{document}

\title{Pricing and Hedging Performance on Pegged FX Markets Based on a Regime Switching Model}

\date{\today}

\author{Samuel Drapeau}
\thanks{National Science Foundation of China, Grants Numbers: 11971310 and 11671257; Grant ``Assessment of Risk and Uncertainty in Finance'' number AF0710020 from Shanghai Jiao Tong University; are gratefully acknowledged.
\\
We also thanks Tan Wang and Tao Wang for fruitful discussions, insights for this problem and help with the data sets.
}
\address{School of Mathematical Sciences \& Shanghai Advanced Institute for Finance (CAFR)\newline Shanghai Jiao Tong University, Shanghai, China}
\email{sdrapeau@saif.sjtu.edu.cn}
\urladdr{http://www.samuel-drapeau.info}

\author{Yunbo Zhang}
\address{School of Mathematical Sciences\newline Shanghai Jiao Tong University, Shanghai, China}
\email{zybsjtu@sjtu.edu.cn}

\begin{abstract}
    This paper investigates the hedging performance of pegged foreign exchange market in a regime switching (RS) model introduced in \citet{drapeau2019}.
    We compare two prices, an exact solution and first order approximation and provide the bounds for the error.
    We provide exact RS delta, approximated RS delta as well as mean variance hedging strategies for this specific model and compare their performance.
    To improve the efficiency of the pricing and calibration procedure, the Fourier approach of this regime-switching model is developed in our work.
    It turns out that: 1 -- the calibration of the volatility surface with this regime switching model outperforms on real data the classical SABR model; 2 -- the Fourier approach is significantly faster than the direct approach; 3 -- in terms of hedging, the approximated RS delta hedge is a viable alternative to the exact RS delta hedge while significantly faster.
\end{abstract}

\maketitle

\section{Introduction}\label{sec:intro}
In pegged markets, unlike free floating ones, the rate of the currency pair is fixed to a given value.
In such a situation, there is no rational need for any option market.
In reality, most modern FX pegged markets are slightly more flexible in the sense that the rate is allowed to fluctuate within a narrow band around a target rate.\footnote{The target rate maybe crawling at a lower frequency too.}
A parade example is the situation of HKDUSD with a target rate of 7.80 HKD to one USD but free to move between 7.75 and 7.85 HKD to one USD, see Figure \ref{fig:hkdeur}.
\begin{figure}[H]
    \centering
    \begin{minipage}{0.495\textwidth}
        \includegraphics[width=\textwidth]{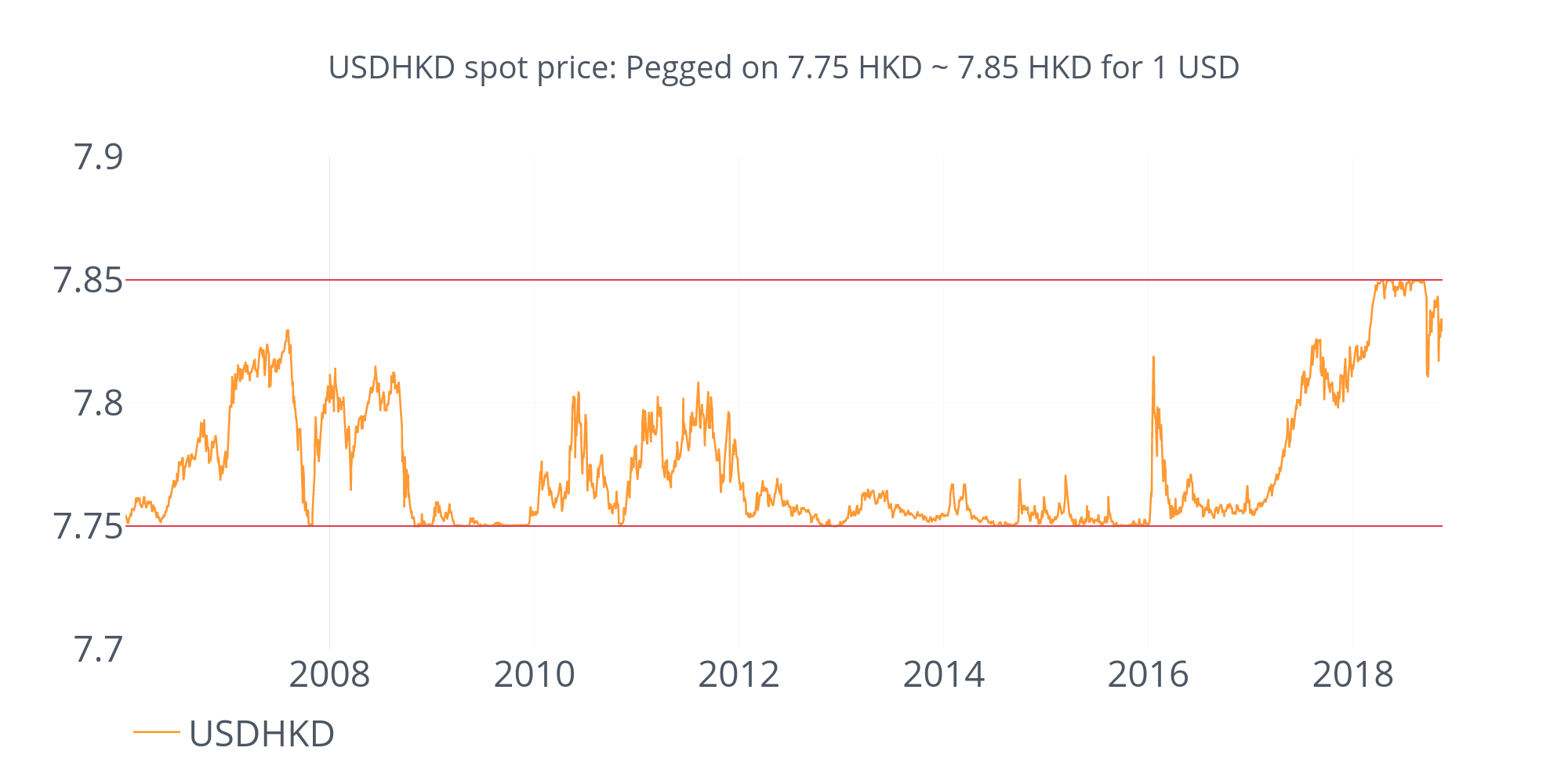}
    \end{minipage}
    \begin{minipage}{0.495\textwidth}
        \includegraphics[width=\textwidth]{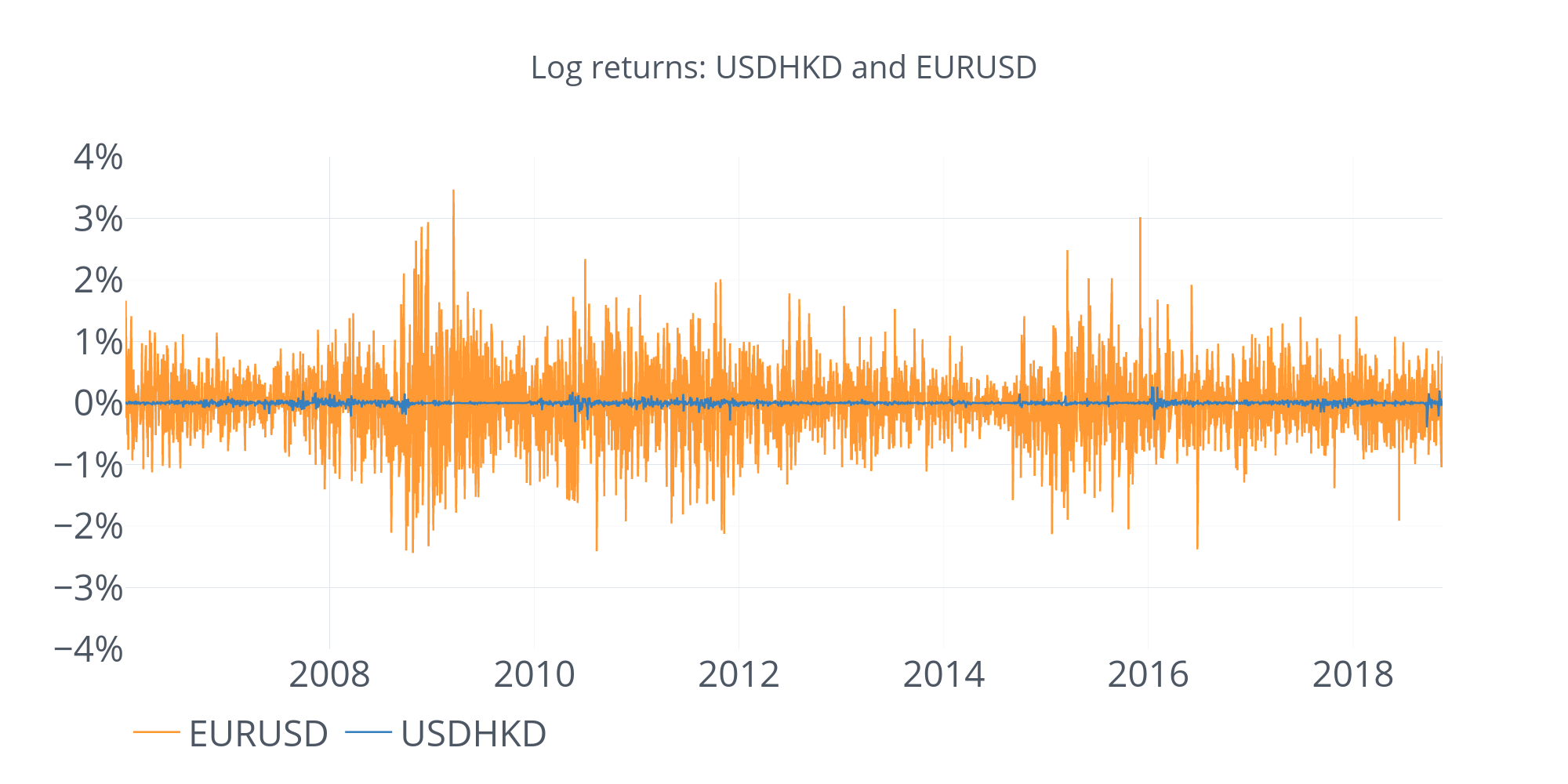}
    \end{minipage}
    \caption{Hong Kong dollar fluctuates between 7.75 and 7.85 HKD to one USD.
    Comparison between the returns on HKDUSD and USDEUR.}
    \label{fig:hkdeur}
\end{figure}
Even in recent time and around the world, pegged FX markets are or have been present: The European Monetary System, Chinese RMB, Thai Bath, Czech Koruna, Hong Kong Dollar, etc.
Even though, most of these markets allow for some fluctuations, the volatility is extremely low as Figure \ref{fig:hkdeur} shows by superposing the volatility of rate returns for HKDUSD in comparison to the free floating EURUSD.
It is therefore counter-intuitive why such markets do have active options trading with puzzling prices.
Indeed, according to the quotations of these markets, many strikes are significantly out of the range with non zero values, see Figure \ref{fig:quots_out_range}.
\begin{figure}[H]
	\centering
	\includegraphics[width=0.7\textwidth]{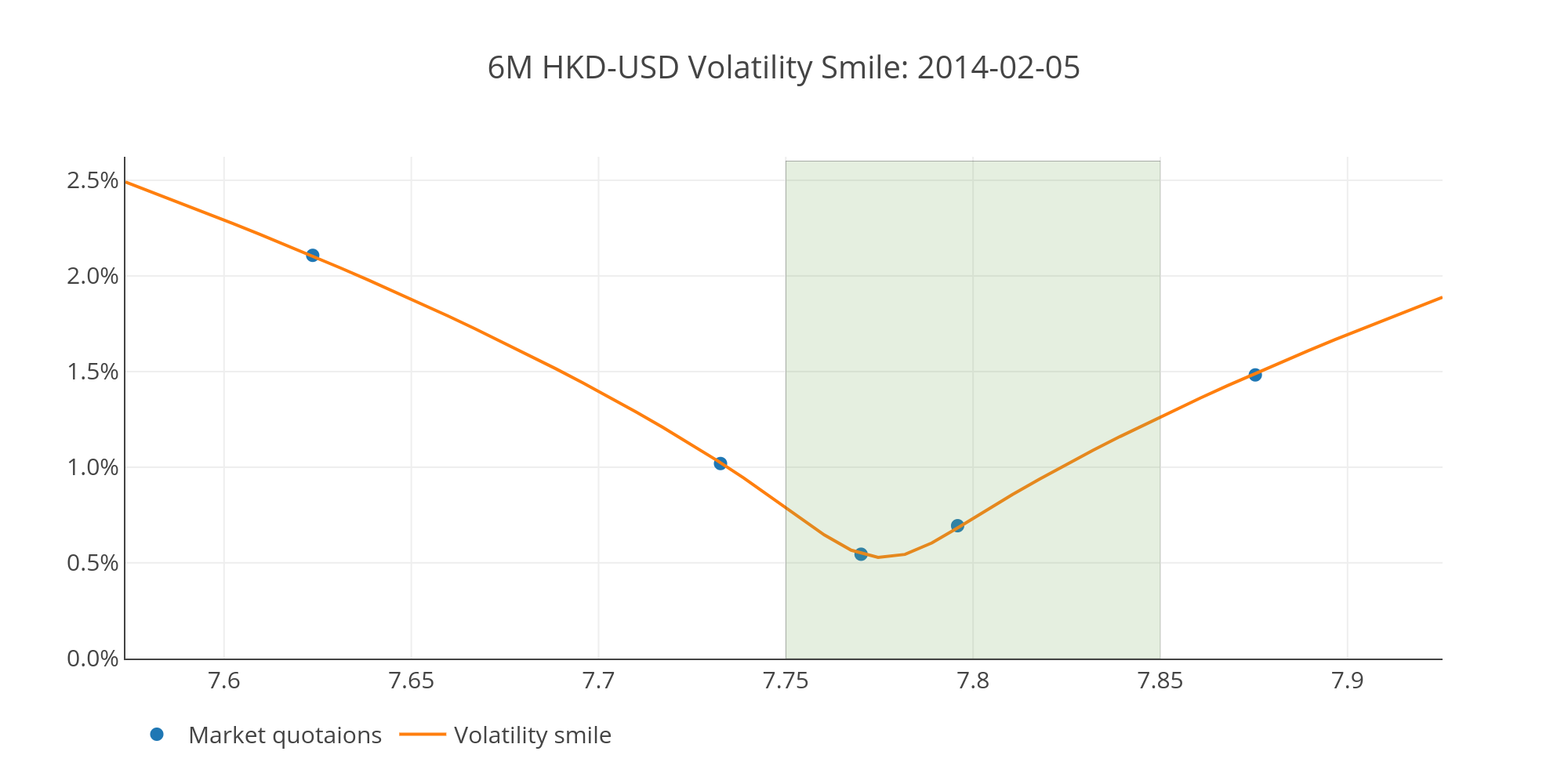}
	\caption{As of 18 May 2005, the Hong Kong Monetary Authority indicated a lower and upper guaranteed limits of 7.75 HKD and 7.85 HKD to one USD, respectively.
	The strikes for $10\%$, $25\%$ delta put and $10\%$ delta call quotations from the market, are out of the band.}
	\label{fig:quots_out_range}
\end{figure}

In a recent paper, \citet{drapeau2019} address this puzzle.
They point as a common intuition for such anomalies that the pegged could eventually be dropped at an uncertain time in the future.
In this context, \citet{drapeau2019} design a regime switching (RS) model to describe this situation, provide a pricing formula under the lens of which they revisit the framework based on a calibration to real data.
We refer to \cite{drapeau2019} for an extended explanation of the pegged currency situation with an empirical analysis of these anomalies, as well as the resulting study of the calibrated model.
Yet, one striking result of this work is that this regime switching model for the pegged FX market provides an excellent calibration accuracy to market quotations.

Building up on this fact, we deepen the study for pricing and hedging of this particular model for pegged markets with a focus on accuracy and performance.
First, from a theoretical perspective, we provide a first order approximation together with error bounds to the true pricing and delta formula.
Second, for computational reasons, we also adopt a Fourier approach by computing the explicit formulation of the moment generating function for this specific model.
Third, from the incompleteness of this model, full hedging is not possible.
We therefore derive an explicit formulation in this model for self-financing mean-variance (MV) hedging strategies.
Based on this theoretical part, we test the hedging performance on the simulated data set at "no-jump" and "jump" scenarios with Black-Scholes (BS) delta, RS delta, approximated RS delta, mean-variance and approximated
mean-variance hedges.
It turns out that mean-variance strategies produce large errors at "no-jump" scenarios, while the improvement is not significant at "jump" scenarios.
In the empirical study part, based on a HKDUSD option data set from 2014-01-01 to 2018-07-12, we compare the performance of each pricer for the calibration of the volatility surface using as a benchmark the classical SABR model.
The exact pricing formula is very accurate and outperforms SABR model almost all the time, while the computational time of the Fourier approach is faster than the direct pricing method by a factor of over 5.
However, the calibration accuracy of approximated RS model is poor.
Given these calibration parameters, we construct the parameter surface for the implied volatility.
Based on this surface, we test BS delta, RS delta, approximated RS delta hedging strategies on the real data set.
We find out that the hedging error of the approximated RS delta hedge does not differ much from the two others, while the computational time of the approximated RS delta is significantly faster than the exact one.

Research literature on pegged market in view of option pricing is particularly scarce, we refer to \citep{drapeau2019} for an overview and discussion about it.
From the quantitative viewpoint however, the following works address the problem of option pricing for bounded diffusions in view of pegged markets, starting with \citet{ingersoll1997} and then \citep{rady1997, carr2016, carr2017}.
These models while self contained for pricing within a band, do not apply for real data calibration on pegged markets as they do not allow for strikes outside of the band.
Starting with \citet{naik1993}, continuous time regime switching processes in the context of option pricing have been widely studied from empirical and theoretical viewpoint, see \citep{hamilton1989, hamilton1990, liu2010,ramponi2011, di2009, boyarchenko2009, shaoyong2017} among others and the references therein.
In contrast to the usual frameworks, the regime switch in this specific model is not independent of the underlying process but triggered by the jump of which.
To our knowledge, the closest form to this model is the one in \cite{shaoyong2017} which is more general.
However, the focus of \cite{shaoyong2017} is on local risk-minimizing hedging portfolio.
As for Fourier and FFT methods in finance, they too have been widely studied starting with \citet{heston1993} and \citep{carr1999}.
We refer to \citep{lewis2001} and \cite{schmelzle2010} for an overview of the literature as well as numerical quadratures comparison and practical implementation details.
To our knowledge, the explicit formula of the moment generating function of this regime switching model is new.
Finally, mean-variance hedging is also standard, starting with \citet{duffie1991} and further \citep{schweizer1992, gourieroux1998, laurent1999, pham1998, arai2005} and the reference therein.
In this paper, we provide an explicit mean-variance hedging formula for this specific regime switching model.

The paper is organized as follows:
After a brief introduction of the model in Section \ref{sec:model}, Section \ref{sec:RSprice_hedge_fourier} is dedicated to the theoretical results where we provide the RS pricing formula, the Fourier approach for RS model, the approximated RS pricing formula with error bounds, as well as the mean-variance hedging strategy.
Section \ref{sec:simulation} present the hedging performance for different strategies on simulated data set. 
Section \ref{sec:calibration} concerns the calibration results of the volatility surface on a real data set. 
Section \ref{sec:real_hedge} provides the hedging performance on real data set from 2014-01-01 to 2018-07-12.
Finally, Section \ref{sec:conclution} concludes.

\section{Model}\label{sec:model}
Let $(\Omega, \mathcal{F}, \mathbb{F},P)$ be a filtrated probability space carrying a one dimensional Poisson random measure $M(dy,ds)$ and a Brownian motion $W$ adapted to a filtration $\mathbb{F}=\{\mathcal{F}_t\colon 0\leq t < \infty\}$ satisfying the usual conditions.
We assume that $M(dy,ds)$ on $\mathbb{R}^+ \times \mathbb{R} \setminus\{0\}$ be a Poisson random measure and denote the corresponding Poisson process as $M(t):= M([0, t], \mathbb{R}\setminus\{0\})$ with the intensity $\lambda$.
The compensator measure of $M(dy,ds)$ is given by $m(dy,ds) = \lambda N^\prime_{u,\delta^2}(y)dyds$ where $N_{u,\delta^2}(\cdot)$ is the CDF of a normal distribution $\mathcal{N}(u,\delta^2)$.
In a pegged foreign exchange market, our regime switching model assumes that the spot price process $S$ follows the dynamic
\begin{equation}\label{equ:my_model}
    \frac{dS(t)}{S(t-)}=(r_d - r_f - \lambda\kappa(\alpha(t-)))dt+\sigma(\alpha(t-))dW(t)+\int_\mathbb{R}\left(e^{\gamma(y,\alpha(t-))}-1\right)M(dy,dt),
\end{equation}
where the constants $r_d$ and $r_f$ represent the domestic and foreign interest rates, respectively.
The two-state continuous time Markov chain $\alpha(t)$ is defined as
\begin{equation*}
    \alpha(t) = 1_{[\tau,\infty)}(t)= M^{\tau}(t), \quad\text{ where }\quad \tau := \inf\{t: M(t) \geq 1\}.
\end{equation*}
In other terms, we have a single regime switch triggered by the first jump of $M$.
The regime switching volatility and jump size $\sigma$ and $\gamma$ are respectively given by
\begin{equation*}
    \sigma(\alpha)=
    \begin{cases}
        \ubar{\sigma}, &\text{ if } \alpha=0\\
        \bar{\sigma}, &\text{ if } \alpha=1
    \end{cases}
    \quad\text{ and }\quad
    \gamma(y,\alpha) =
    \begin{cases}
        y, &\text{ if } \alpha=0\\
        0, &\text{ if } \alpha=1
    \end{cases}
\end{equation*}
where $0<\ubar{\sigma}\leq \bar{\sigma}$.
The resulting compensator is given by $\kappa(0) = e^{u+\delta^2/2}-1$ and $\kappa(1) = 0$.
For ease of notations, throughout we use the notation $\kappa:= e^{u+\delta^2/2}-1$, average percentual change of spot price right after the switch.

We assume that $W(\cdot)$, $M(\cdot)$ and $Y$ are independent of each other.
According to Dol\'eans-Dade exponential, it follows that
\begin{equation*}
    S(t) = S_0\exp\left((r_d - r_f)t + X(t)\right)
\end{equation*}
where $S_0$ denotes the spot price at time $0$ and 
\begin{multline*}
    X(t) = -\int_0^t\left(\frac{\sigma^2(\alpha(s-))}{2}+\lambda\kappa(\alpha(s-))\right)ds
    +\int_0^t\sigma(\alpha(s-))dW(s)\\
    +\int_0^t\int_\mathbb{R}\gamma(y,\alpha(s-))M(dy,ds).
\end{multline*}
In particular, $\tilde S(t):= S(t)/e^{(r_d - r_f)t}$ is a martingale under $P$ which is therefore a risk neutral measure.

\begin{remark}
    Note that, under this model, due its diffusive nature, the price might break the band before the drop of the peg.
    As mentioned in the introduction, it is possible to model diffusion constrained within bands that can be used as model before this regime switch.
    This would however result in more complex formulations in terms of pricing and hedging as well as problematic calibrations and implementations.
    In contrast, as explained in \citep{drapeau2019}, the present model is intuitive in its framework and easy to implement and calibrate while at the costs of non constrained price evolution before the drop.
    However, assuming a negligible volatility $\ubar{\sigma}$ -- as in the empirical facts -- the resulting probability to break the band significantly before the regime switch is therefore almost zero.
    This is the reason why we adopt such a model as a proxy for the modelling of pegged markets since our focus is on the effect of the market beliefs for a future drop.
\end{remark}

\section{Pricing and hedging}\label{sec:RSprice_hedge_fourier}
Throughout, we denote by $BS$ the Black and Scholes price of call options in the Black-Scholes setting where the spot price follows a geometric Brownian motion.
For a call with maturity $T$ and strike $K$, this price is given by
\begin{equation}\label{equ:gk}
    BS(S_0,\sigma \sqrt{T}) = e^{-r_fT}S_0N(d_+) - e^{-r_dT}KN(d_-),
\end{equation}
see \citet{garman1983}, where $N$ denotes the cumulative distribution function of a standard normal distribution and $d_\pm :=  \ln(S_0e^{(r_d-r_f) T}/K)/(\sigma\sqrt{T}) \pm \sigma\sqrt{T}/2$.
Let further $\Delta_{BS}$ and $\mathcal{V}_{BS}$ denote the corresponding pips spot delta and vega , that is
\begin{align}
    \Delta_{BS}(S_0,\sigma\sqrt{T}) = \frac{\partial BS}{\partial S_0} = e^{-r_f T}N(d_+) \label{equ:gkdelta}\\
    \mathcal{V}_{BS}(S_0,\sigma\sqrt{T}) = \frac{\partial BS}{\partial \sigma} = S_0e^{-r_f T}N^\prime(d_+)\sqrt{T} \label{equ:gkvega}
\end{align}

Finally, we denote by $\sigma_{BS}(K)$ the implied volatility of an observed price $\pi(K)$ with strike $K$ under this model -- for a given maturity -- that is $\sigma_{BS}(K) := \left( \sigma\mapsto BS(\sigma) \right)^{-1}(\pi(K))$.

\subsection{Martingale approach}\label{subsec:directRS}
Taking $P$ as the risk neutral measure, we have the following pricing formula under the model \eqref{equ:my_model}, see also \citet{drapeau2019}.

\begin{proposition}
    In the regime switching jump diffusion model \eqref{equ:my_model}, the price of a European call option with parameter $\theta=(\ubar{\sigma}, \bar{\sigma}, \lambda, u, \delta)$ is given by
    \begin{multline}\label{equ:rs}
        V\left( \theta\right) = BS\left(S_0e^{-\lambda \kappa T},K,T,r_d,r_f,\ubar{\sigma}\sqrt{T}\right)e^{-\lambda T} \\
        + \int_0^T BS\left(S_0e^{-\lambda \kappa t}(1+\kappa),K,T,r_d,r_f,\sqrt{\ubar{\sigma}^2t + \bar{\sigma}^2(T-t)}\right)\lambda e^{-\lambda t}dt 
    \end{multline}
    Further, the delta-hedging strategy is given by
    \begin{multline}\label{equ:rsdelta}
        \Delta\left(\theta\right) = \Delta_{BS}\left(S_0e^{-\lambda \kappa T},K,T,r_d,r_f,\ubar{\sigma}\sqrt{T}\right)e^{-\lambda T(1+\kappa)} \\
        + \int_0^T \Delta_{BS}\left(S_0e^{-\lambda \kappa t}(1+\kappa),K,T,r_d,r_f,\sqrt{\ubar{\sigma}^2t + \bar{\sigma}^2(T-t)}\right)\lambda e^{-\lambda t(1+\kappa)}(1+\kappa)dt.
    \end{multline}
\end{proposition}
By $\sigma_{BS}(K; \theta)$ we denote the parametrized regime switching implied volatility for prices $V(\theta)$ with parameter $\theta$, that is
\begin{equation}\label{equ:sigma_parameter}
    \sigma_{RS}(K;\theta) := \left( \sigma \mapsto BS(\sigma) \right)^{-1}(V(K; \theta)).
\end{equation}
Figure \ref{fig:sensitivity_paras} shows how the parameters effect the shape of the volatility smile.
For the illustration, we assume the standard deviation parameter $\delta = 0$, $S_0 = 7.77$, $r_d = r_f $, $T=0.5$ and strike prices $K$ in $[7.5, 8]$.
\begin{figure}[H]
    \centering
    \begin{minipage}{0.495\textwidth}
        \includegraphics[width=\textwidth]{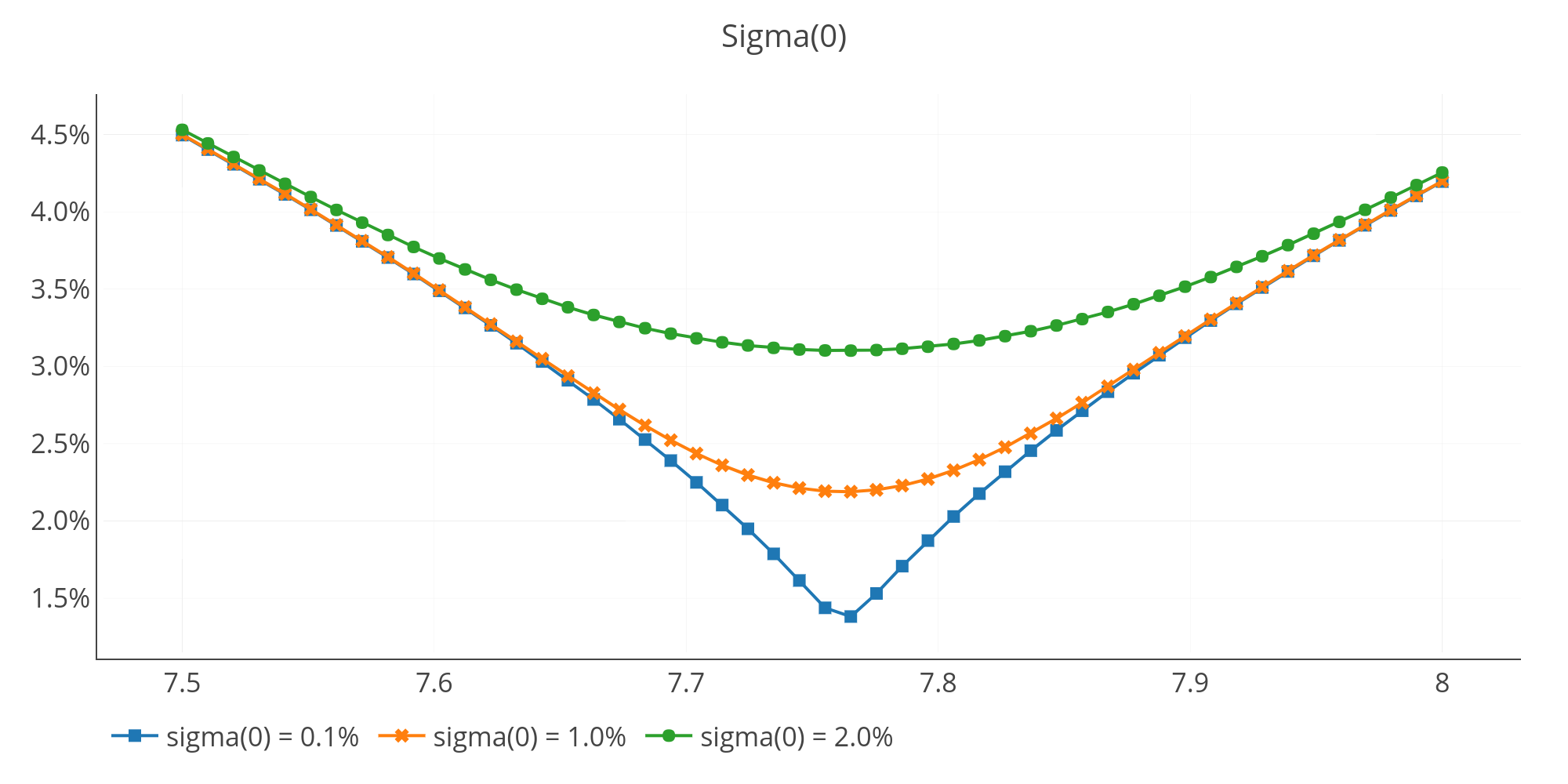}
    \end{minipage}
    \begin{minipage}{0.495\textwidth}
        \includegraphics[width=\textwidth]{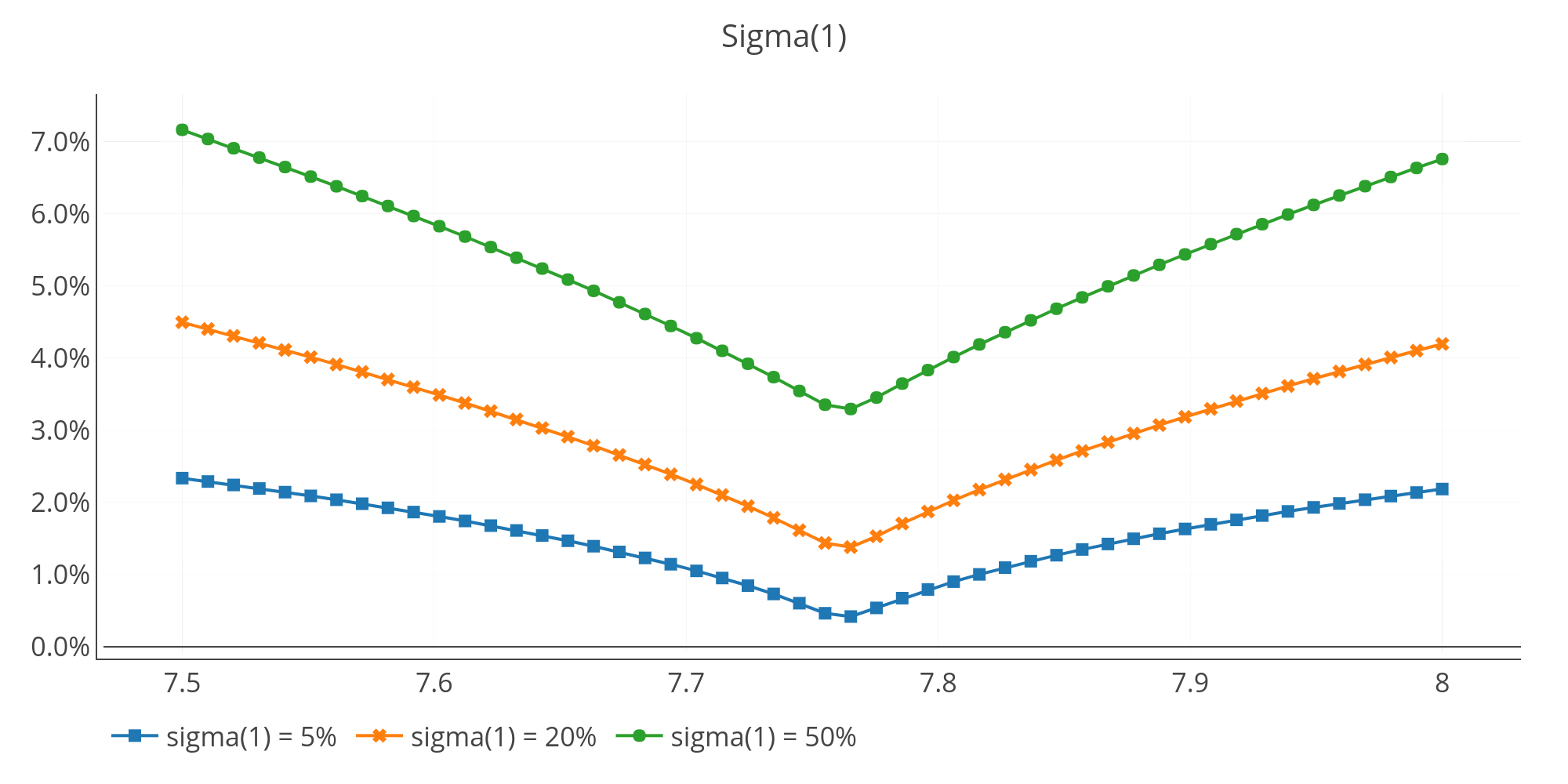}
    \end{minipage}

    \begin{minipage}{0.495\textwidth}
        \includegraphics[width=\textwidth]{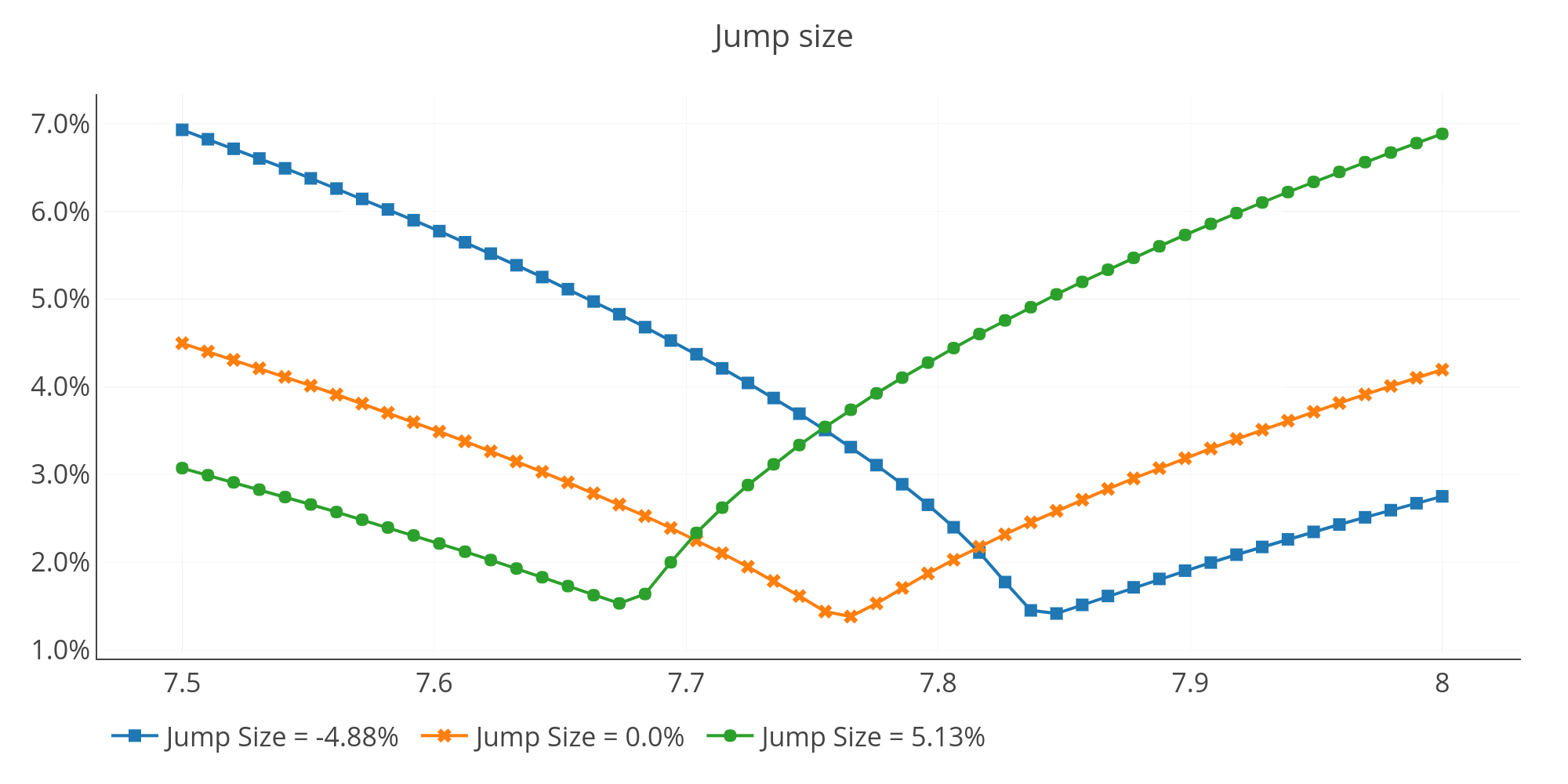}
    \end{minipage}
    \begin{minipage}{0.495\textwidth}
        \includegraphics[width=\textwidth]{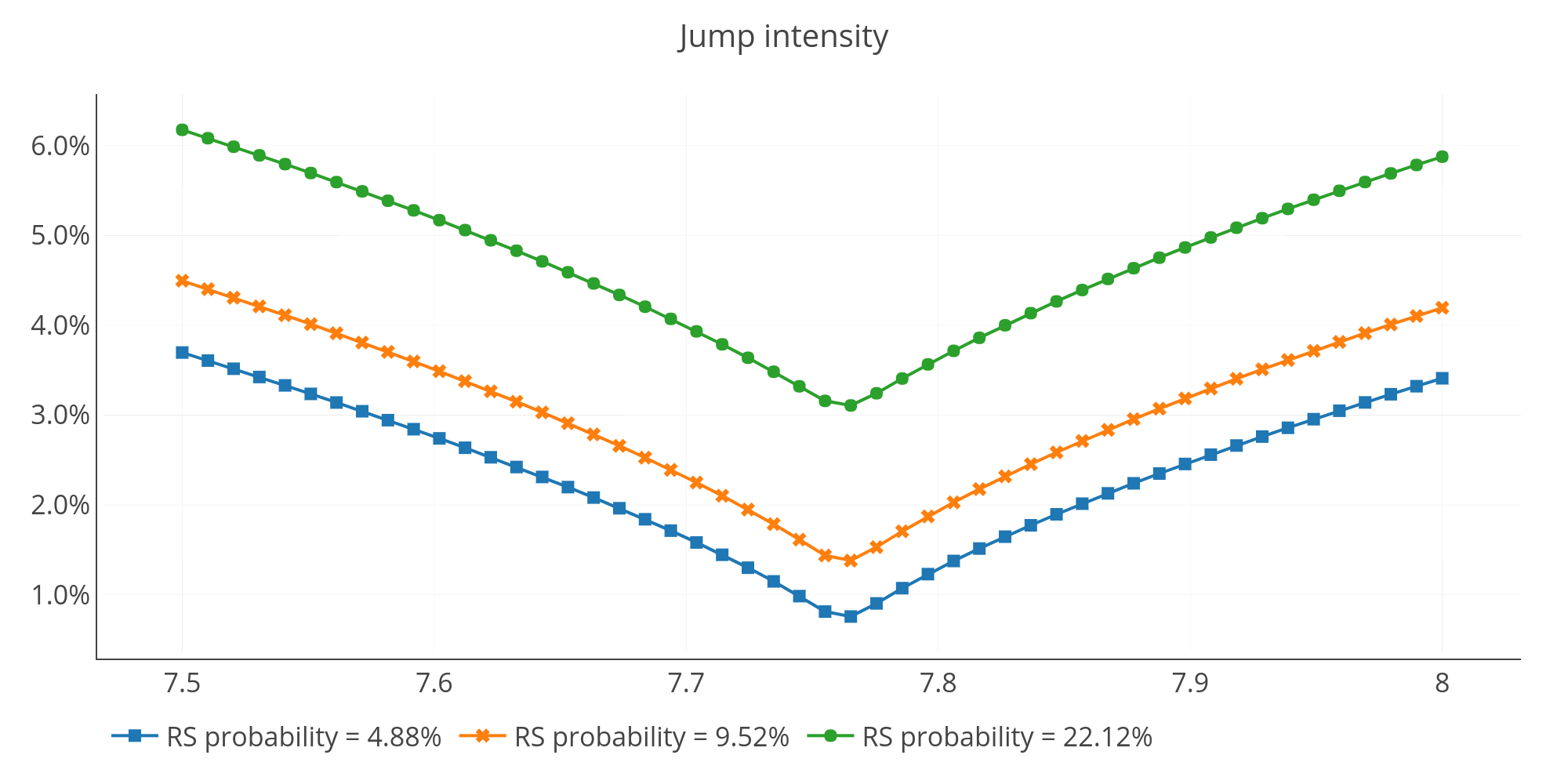}
    \end{minipage}
    \caption{Implied volatility: sensitivity to parameters}
    \label{fig:sensitivity_paras}
\end{figure}
\begin{remark}
    Note that this regime switching implied volatility does rather generate some ``bird'' shaped smiles usually observed in pegged markets, see \citep{drapeau2019}.
    Both jump intensity $\lambda$ and upper volatility $\bar{\sigma}$ do translate upwards the volatility smile.
    However, the upper volatility level also tends to ``bend'' the smile as it increases.
    Finally, the jump direction $u$ influences the skew.
\end{remark}

\subsection{Fourier approach}\label{subsec:Fourier}
Following a Fourier approach to option pricing, see \citet{lewis2001}, we denote by 
\begin{equation*}
    \varphi_T(z;\theta) := E\left[e^{i z X(T)}\right]  
\end{equation*}
the characteristic function of $X(T)$.
Assuming $0<v<1$, the call option price is given by
\begin{equation*}
    V\left(\theta\right) = S_0e^{-r_fT} - \frac{Ke^{-r_dT}}{2\pi}\int_{iv-\infty}^{iv+\infty}e^{-izk}\varphi_T(-z;\theta)\frac{1}{z^2-iz}dz
\end{equation*}
where $k=\ln(S_0/K) + (r_d-r_f)T$.
For example, with $v=\frac{1}{2}$, the price reads as
\begin{equation}\label{equ:Fourier_call}
    V\left(\theta\right) = S_0e^{-r_fT} - \frac{\sqrt{S_0K}e^{-(r_d+r_f)T/2}}{\pi}\int_0^\infty \text{Re}\left(e^{iuk}\varphi_T \left(u-\frac i2;\theta\right)\frac{1}{u^2+\frac14}\right)du.
\end{equation}
where $\text{Re}(x)$ denotes the real part of $x$, see \citet{hilpisch2015} for details.
The following proposition provides an explicit expression for the characteristic function $\varphi_T(z;\theta)$, which is up to our knowledge new, since the jump of the model is the trigger for the regime switch and not, as usually assumed, independent of it.
\begin{proposition}
    In regime switching model (\ref{equ:my_model}), the characteristic function $\varphi_T(z;\theta)$ of $X(T)$ is given by
    \begin{equation*}
        \varphi_T(z;\theta)
        = e^{-\lambda T}\phi_{0,T}(z)
        +\lambda e^{izu-z^2\frac{\delta^2}{2}}\frac{e^{-\lambda T}\phi_{0,T}(z) - \phi_{1,T}(z)}{\left(iz+z^2\right)\frac{\bar{\sigma}^2 - \ubar{\sigma}^2}{2} - \lambda(izk+1)}.
    \end{equation*} 
    where
    \begin{equation*}
        \phi_{0,T}(z) = e^{-iz \frac{\ubar{\sigma}^2T}{2} -z^2\frac{\ubar{\sigma}^2T}{2} - iz\lambda\kappa T}
        \quad \text{and}\quad
        \phi_{1,T}(z) = e^{-iz \frac{\bar{\sigma}^2T}{2}-z^2\frac{\bar{\sigma}^2 T}{2}}
    \end{equation*}
\end{proposition}
\begin{proof}
    See Appendix \ref{app:Fourier}.
\end{proof}
\begin{remark}
    The FFT algorithm can be applied to \eqref{equ:Fourier_call}, the details of which are deferred to the Appendix \ref{app:FFT}.
    FFT is not applied in the following empirical study as FFT algorithm turns out not to be time saving for the required accuracy.
\end{remark}

\subsection{A first-order approximation}\label{subsec:approxRS}
The exact solution given by relation \eqref{equ:rs} for the price of a call option in the regime switching model
can be seen as a convex combination between the GK price with low volatility and the average between $0$ and $T$ with high volatilities and jumped price.
This integration part is however computational costly, either in terms of direct computation or using Fourier methods.
Yet, we obtain a straightforward approximation in terms of convex combinations between the two regimes for price and delta as follows
\begin{equation}\label{equ:approx_RS}
    V_{approx}(\theta)   = pBS\left(S_0e^{-\lambda \kappa T}, \ubar{\sigma}\sqrt{T}\right) + (1-p) BS\left( S_0(1+\kappa), \bar{\sigma}\sqrt{T} \right) 
\end{equation}
\begin{equation}\label{equ:approx_RSdelta}
\begin{split}
    \Delta_{approx}(\theta) & = p e^{-\lambda \kappa T}\Delta_{BS}\left(S_0e^{-\lambda \kappa T},\ubar{\sigma}\sqrt{T}\right)	\\
                            & \quad \quad \quad \quad+ (1-p)(1+\kappa) \Delta_{BS}\left( S_0(1+\kappa), \bar{\sigma}\sqrt{T} \right),
\end{split}
\end{equation}
where $p=e^{-\lambda T}$.
\begin{proposition}\label{prop:approx}
    The spot adjusted error between the approximated and exact formulation is given by
    \begin{equation*}
        \frac{\left| V(\theta)-V_{approx}(\theta) \right|}{S_0}  \leq (1-p)\sqrt{\frac{T}{2\pi}}(\bar{\sigma} - \ubar{\sigma}) + |\kappa|(1-p) - p\left|e^{-\lambda\kappa T} - 1\right|.
    \end{equation*}
\end{proposition}
\begin{proof}
    See Appendix \ref{app:first_order_approx}.
\end{proof}

While the exact solution is analytical for the pricing and hedging in the regime switching model (\ref{equ:my_model}), the approximated solution provides a faster method for pricing and hedging.
The drawbacks of this approximation method are the accuracy issues of the resulting option price, accuracy that depends on the value of $\theta$ as well as the maturity.
We provide the relative error of the approximated RS price ($\text{Relative Error}_{V}$) and the approximated RS Delta ($\text{Relative Error}_{\Delta}$) by  
\begin{equation*}
    \text{Relative Error}_{V} = \frac{V_{approx} - V}{S_0}*100\%
    \quad\text{and}\quad
    \text{Relative Error}_{\Delta} = \frac{\Delta_{approx} - \Delta}{S_0}*100\%,
\end{equation*}
respectively.
Figure \ref{fig:relerror_approx_price05} shows the relative error between the exact and approximated RS price (left) and delta (right) with a spot price $S_0$ in $[50,150]$, strike price $K=100$, domestic interest rate $r_d=2 \%$, foreign interest rate $r_f=3\%$, a maturity $T = 1$, lower volatility level $\ubar{\sigma}=2\%$, upper volatility level $\bar{\sigma}=10\%$, jump intensity $\lambda=0.1$ and jump size parameters $u=0.05$ and $\delta=0$. 
\begin{figure}[H]
    \centering
    \begin{minipage}{0.495\textwidth}
        \includegraphics[width=\textwidth]{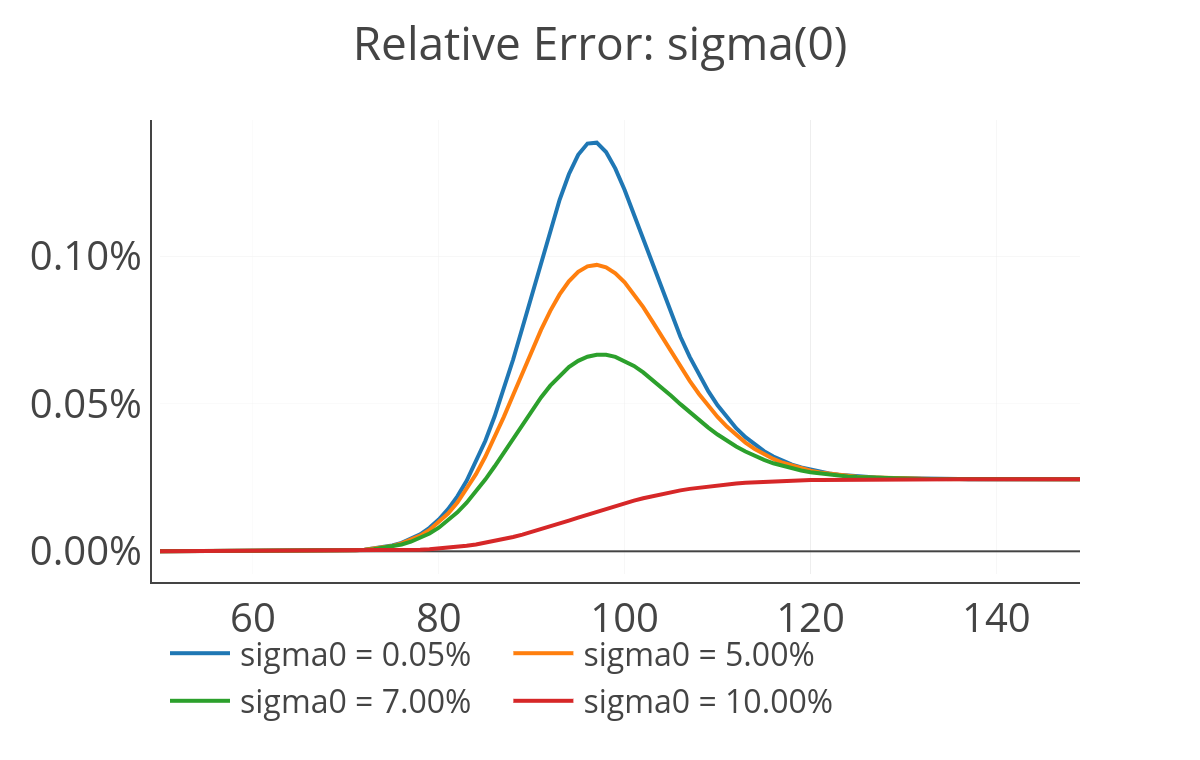}
    \end{minipage}
    \begin{minipage}{0.495\textwidth}
        \includegraphics[width=\textwidth]{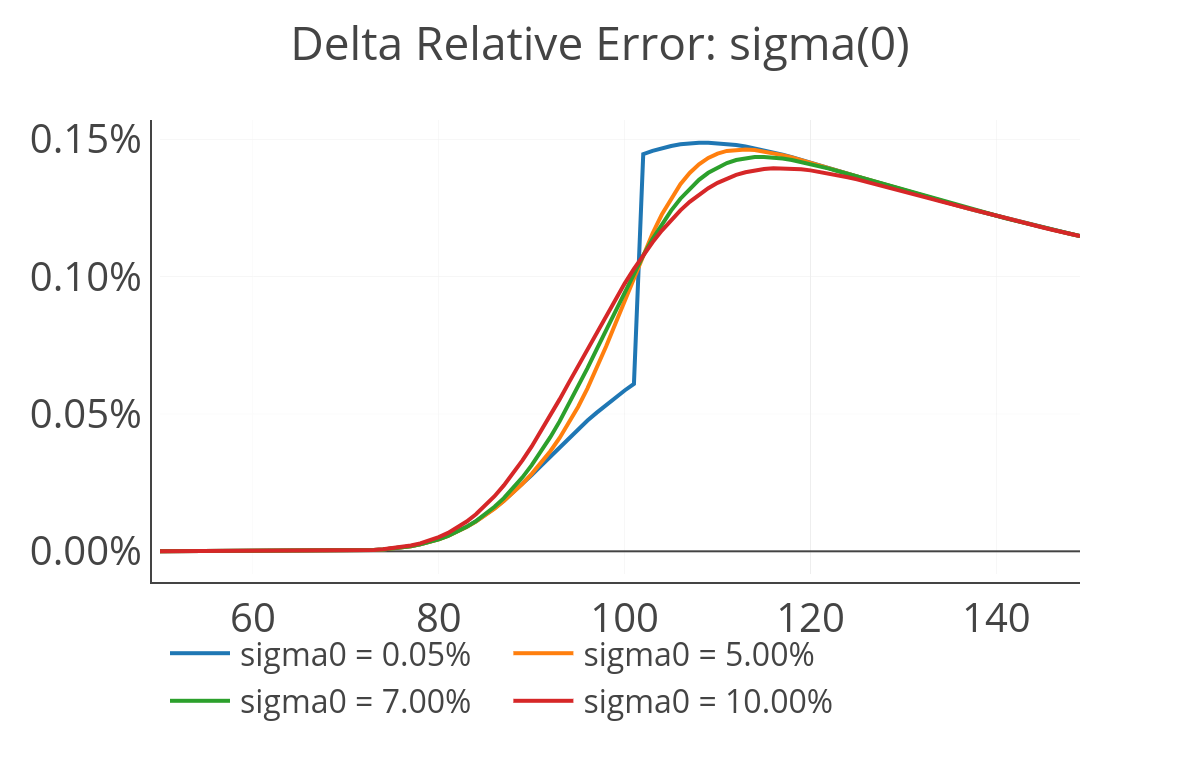}
    \end{minipage}
    \label{fig:relerror_approx_price01}
\end{figure}
\begin{figure}[H]
    \centering
    \begin{minipage}{0.495\textwidth}
        \includegraphics[width=\textwidth]{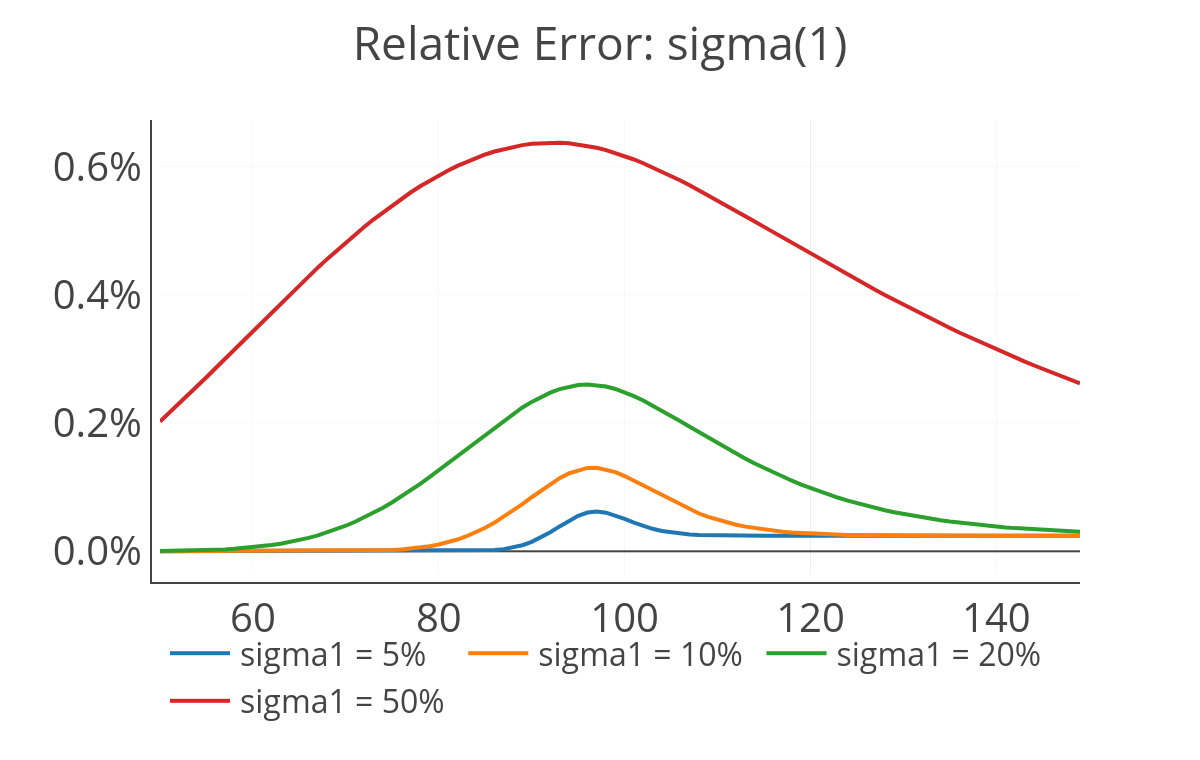}
    \end{minipage}
    \begin{minipage}{0.495\textwidth}
        \includegraphics[width=\textwidth]{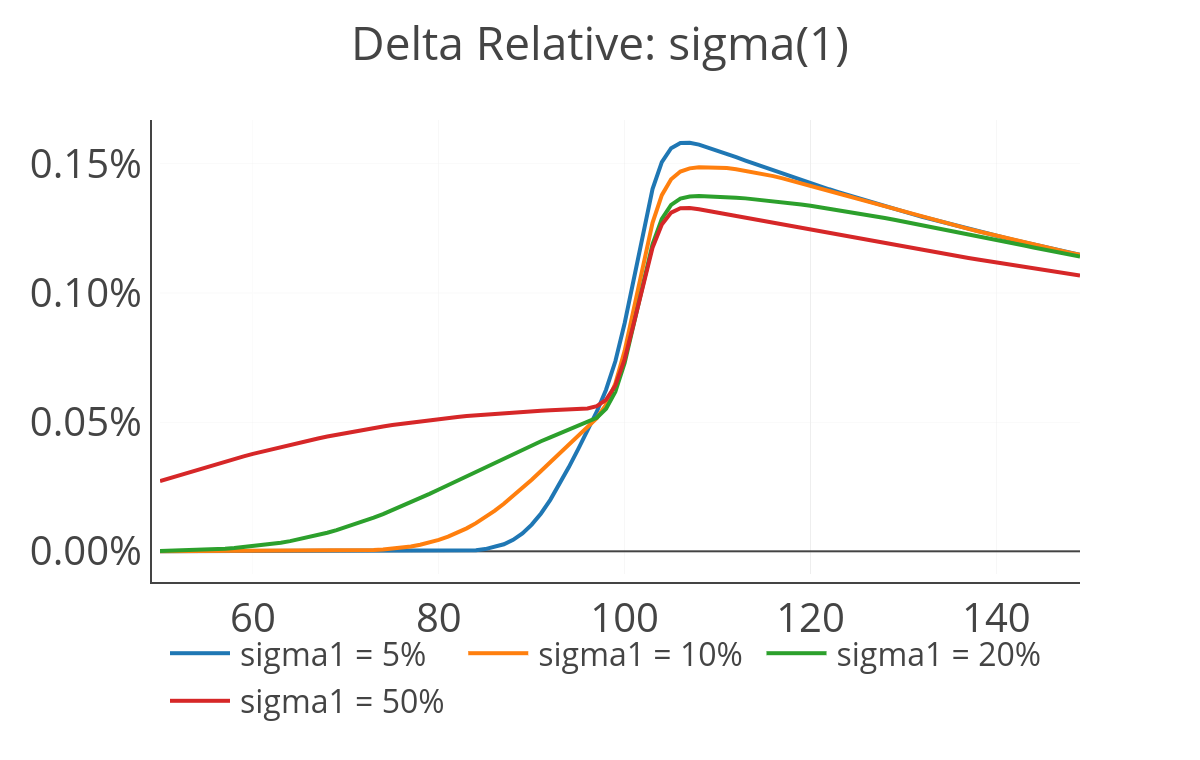}
    \end{minipage}
    \label{fig:relerror_approx_price02}
\end{figure}
\begin{figure}[H]
    \centering
    \begin{minipage}{0.495\textwidth}
        \includegraphics[width=\textwidth]{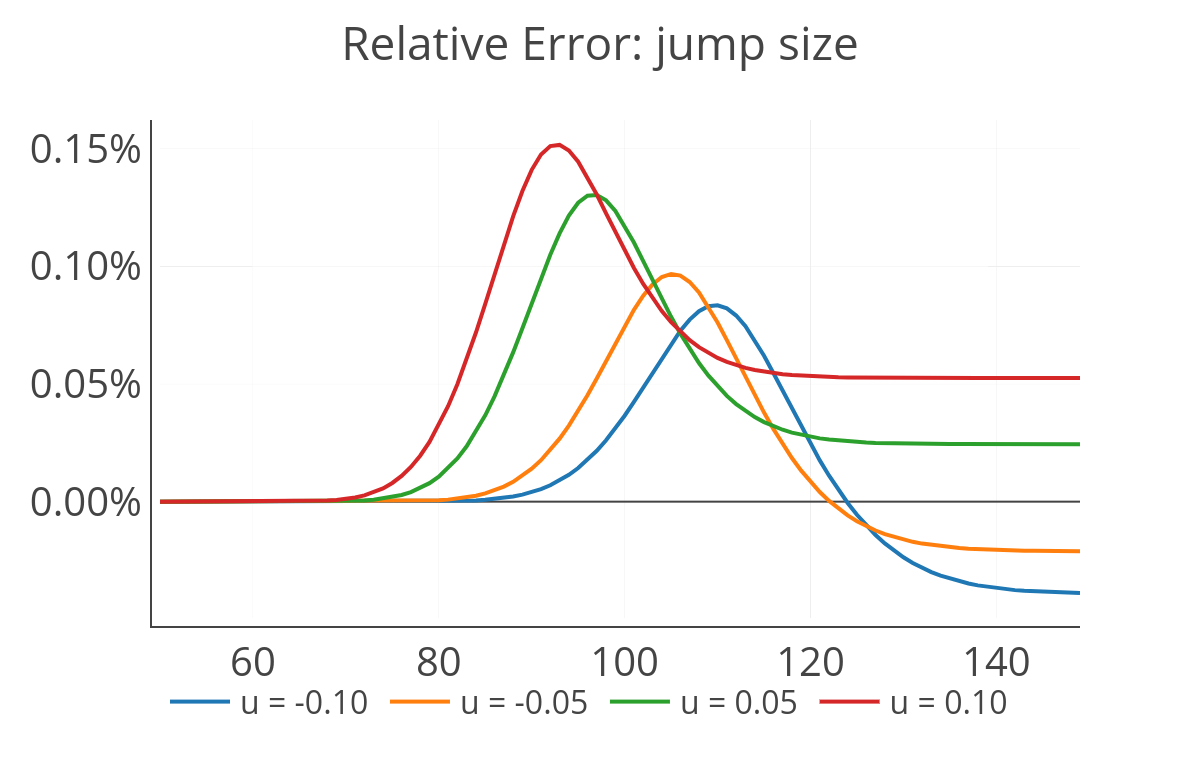}
    \end{minipage}
    \begin{minipage}{0.495\textwidth}
        \includegraphics[width=\textwidth]{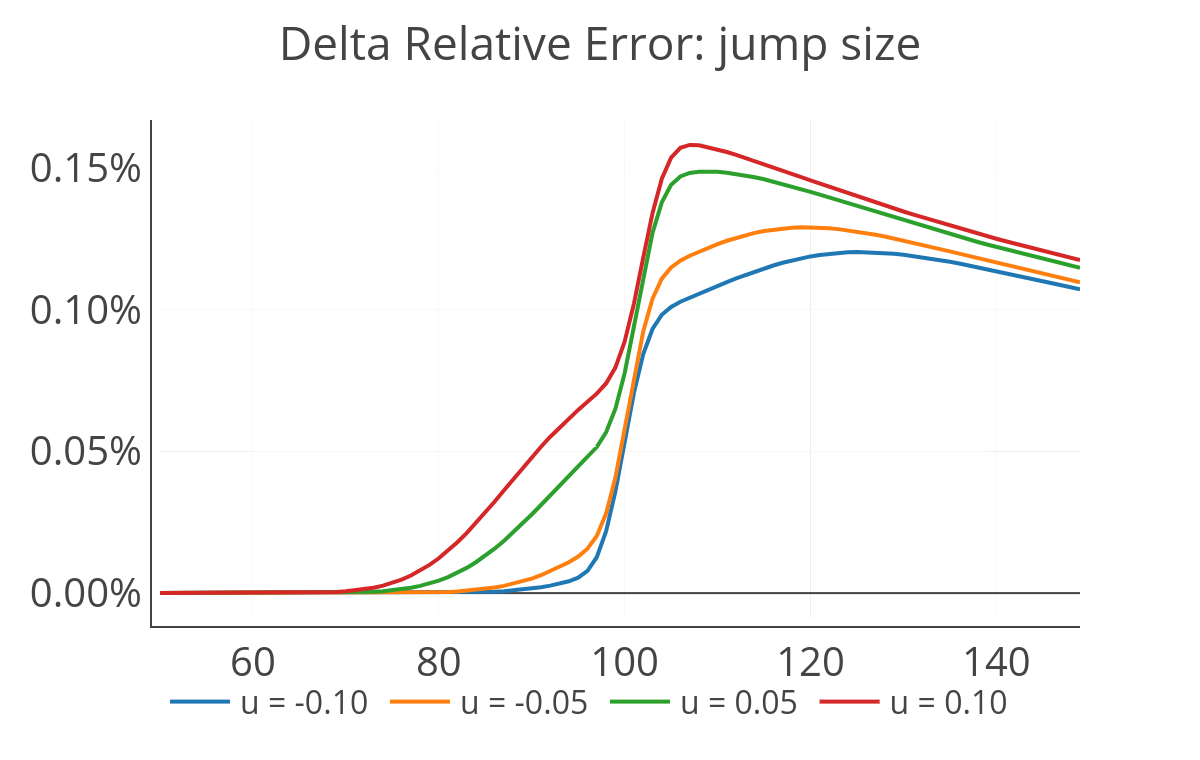}
    \end{minipage}
    \label{fig:relerror_approx_price03}
\end{figure}
\begin{figure}[H]
    \centering
    \begin{minipage}{0.495\textwidth}
        \includegraphics[width=\textwidth]{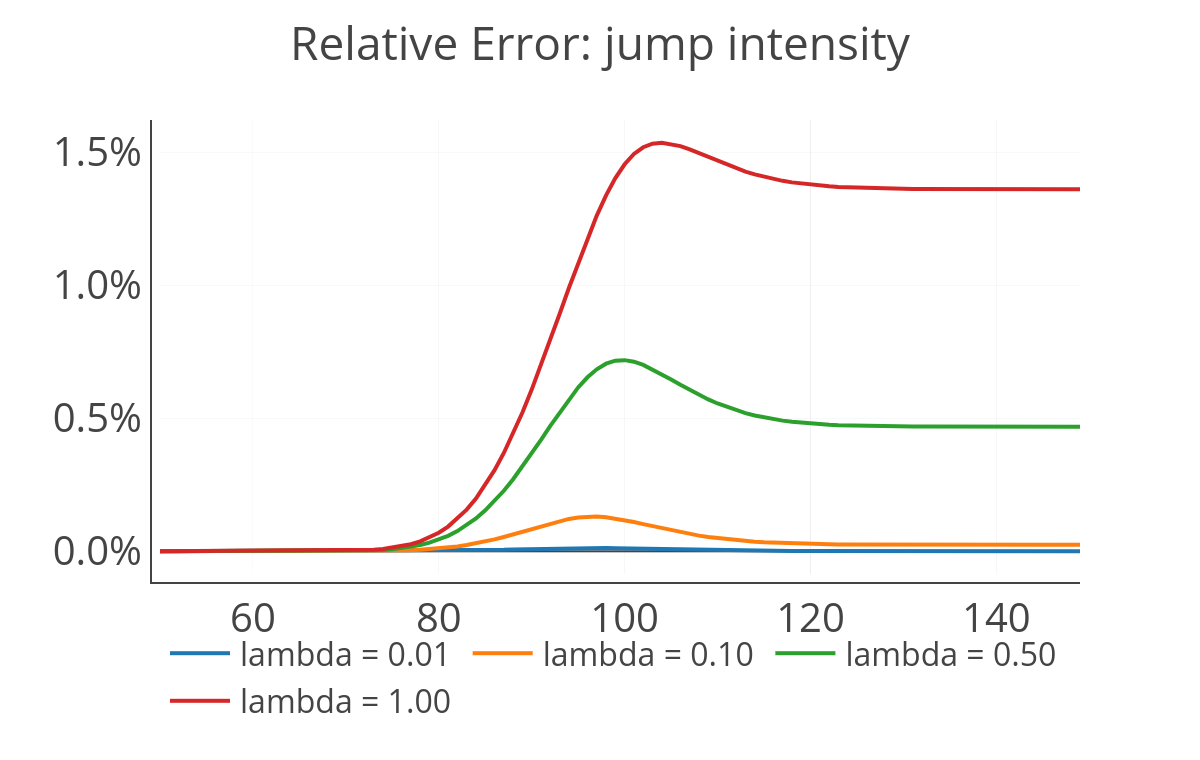}
    \end{minipage}
    \begin{minipage}{0.495\textwidth}
        \includegraphics[width=\textwidth]{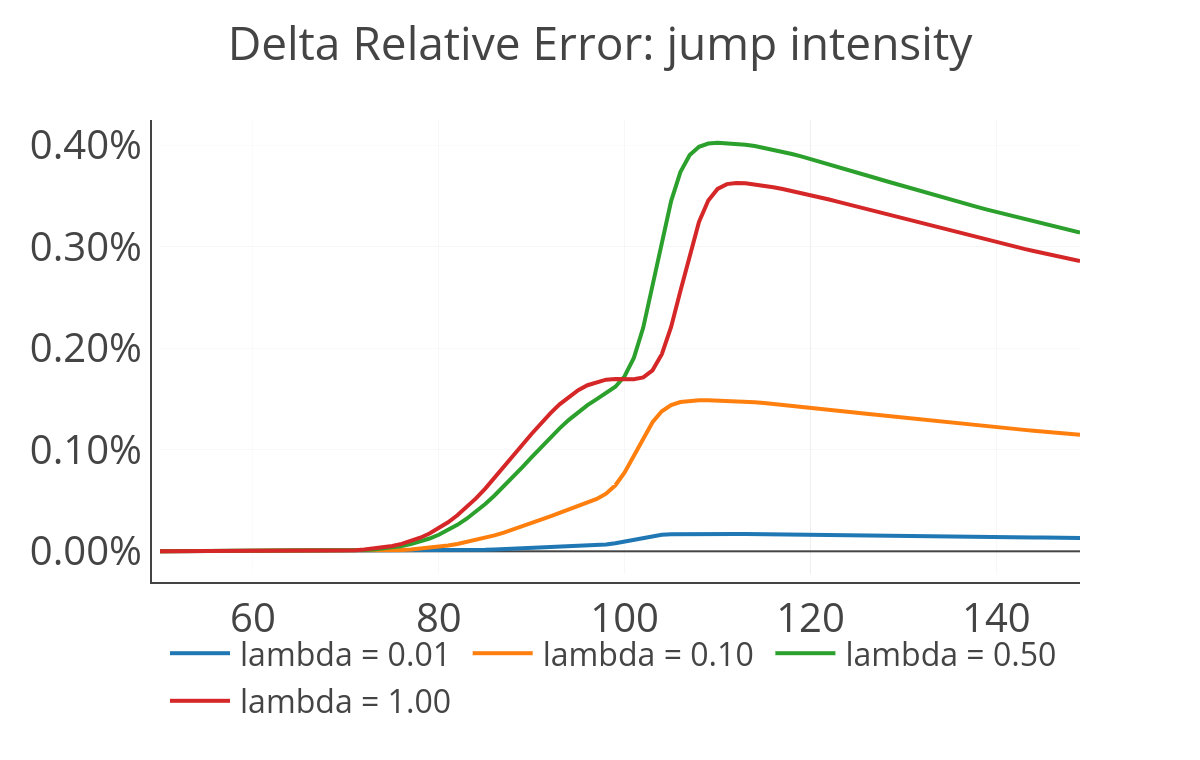}
    \end{minipage}
    \label{fig:relerror_approx_price04}
\end{figure}

\begin{figure}[H]
    \centering
    \begin{minipage}{0.495\textwidth}
        \includegraphics[width=\textwidth]{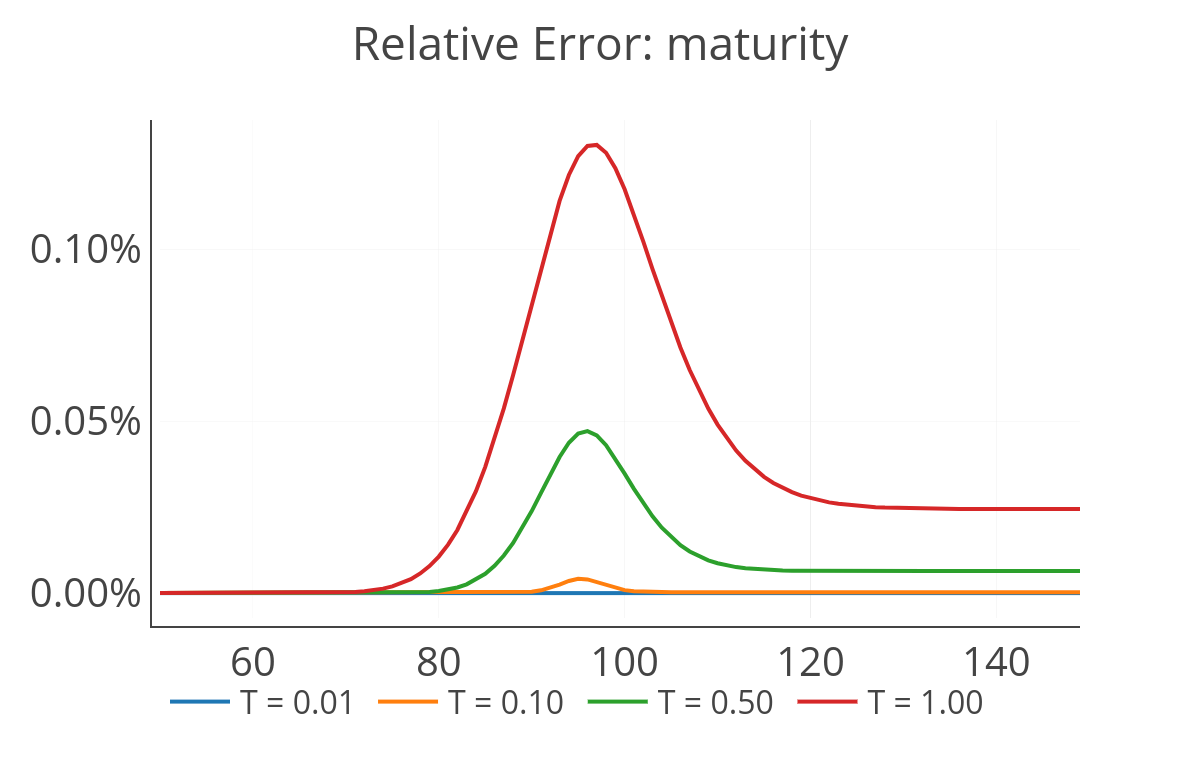}
    \end{minipage}
    \begin{minipage}{0.495\textwidth}
        \includegraphics[width=\textwidth]{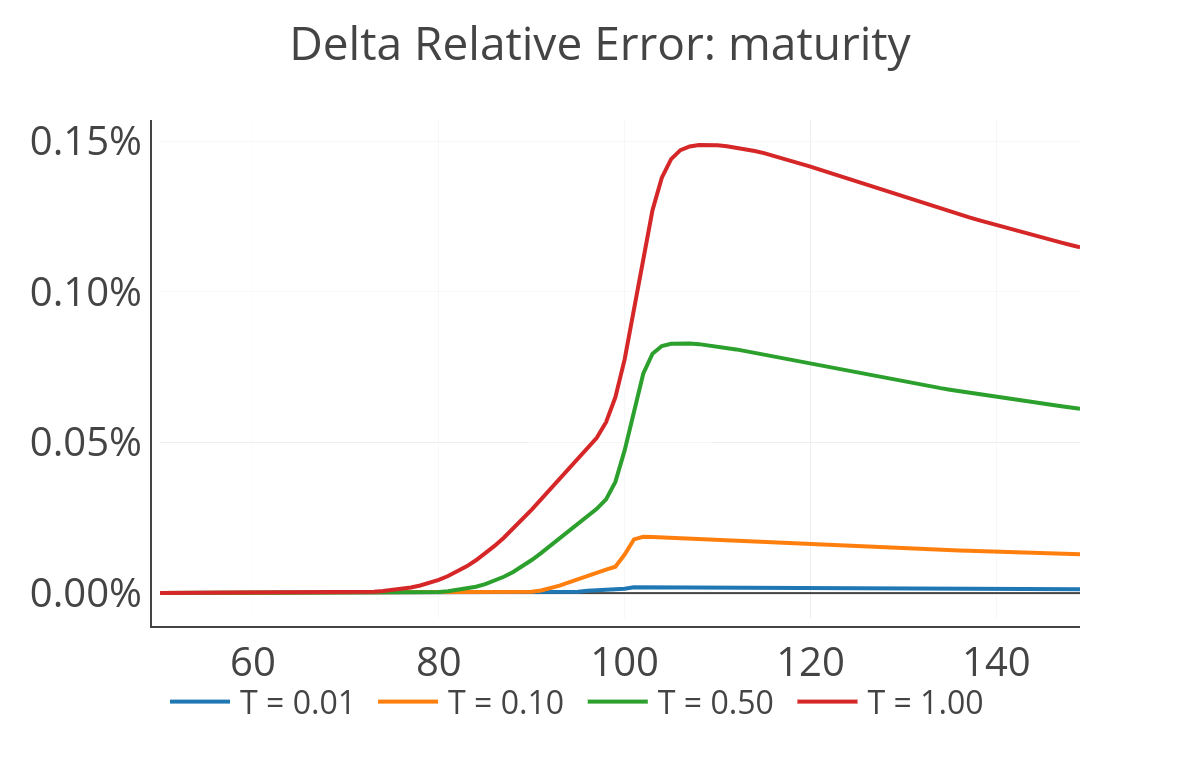}
    \end{minipage}
    \caption{Relative errors as a function of $\underline{\sigma}$, $\bar{\sigma}$, jump size, jump intensity and maturity, respectively.}
    \label{fig:relerror_approx_price05}
\end{figure}

\begin{remark}
    Note that both relative errors are getting large around	at-the-money area.
    The errors are small when the call is out of the money.
    Furthermore, a negative jump size $u$ gives a negative relative error of prices.
\end{remark}

\subsection{Quadratic hedging}\label{subsec:MVhedge}
Due to the assumption of regime switching and jump in the model, the market is not complete, and hence there is no perfect hedging.
A popular self-financing hedging approach is the \textit{mean-variance hedging} where a strategy $\pi$ with initial capital $\pi^0$ is searched so as to minimize the mean squared hedging error
\begin{equation}\label{equ:func_c}
    \inf_{\pi^0,\pi}E\left[\left(\pi^0 - \int_0^T\pi d\tilde S - \tilde H\right)^2\right]
\end{equation}
over all admissible trading strategies.
Here, 
\begin{equation*}
    \tilde H:= H/e^{(r_d - r_f)T}\quad \text{and} \quad \tilde S(t) := S(t)/e^{(r_d - r_f)t}
\end{equation*}
are the discounted payoff and spot price process, respectively.
Notice that, for European call option, the payoff is $H(S(T)) = (S(T)-K)^+$.
Moreover, we define
\begin{equation*}
    \mathcal{C}(t,s,x) := e^{-(r_d-r_f)(T-t)}E[H(S(T))|S_t=s,\alpha(t)=x],\quad x\in\{0,1\}
\end{equation*}
and the discounted value of $\mathcal{C}(t,s,x)$ by
\begin{equation*}
    \tilde {\mathcal{C}}(t,s,x) = e^{-(r_d-r_f)t}{\mathcal{C}}(t,s,x).
\end{equation*}
The following theorem provides the mean-variance hedging strategy in the regime switching model.	
\begin{theorem}\label{thm:mv_hedging}
    The mean-variance hedging portfolio for European currency options is given by
    \begin{equation}\label{equ:mvdelta}
        \pi(t) =
        \begin{cases}
            \displaystyle \frac{\displaystyle \ubar{\sigma}^2\frac{\partial \tilde{\mathcal{C}}}{\partial S}(t,S(t-),0) + \frac{\lambda(e^{u}-1)}{S(t-)}\left[\tilde{\mathcal{C}}(t,S(t-)e^{u+\frac{\delta^2}{2}},1) - \tilde{\mathcal{C}}(t,S(t-),0)\right]}{\displaystyle \ubar{\sigma}^2+\lambda(e^{u+\frac{\delta^2}{2}}-1)^2} \\
            \quad \quad \quad \quad\quad \quad \quad \quad \quad \quad \quad \quad\quad \quad \quad \quad \quad \quad \quad \quad \quad\quad \quad \quad \quad \text{ if } \alpha(t-)=0\\
            \displaystyle \frac{\partial \tilde{\mathcal{C}}}{\partial S}(t,S(t-),1) \\
            \quad \quad \quad \quad\quad \quad \quad \quad\quad \quad \quad \quad\quad \quad \quad \quad \quad \quad \quad \quad \quad\quad \quad \quad \quad \text{ if } \alpha(t-)=1
        \end{cases}
    \end{equation}
    and the initial capital $\pi^0(0) = V(S_0;\theta)$.	
\end{theorem}
\begin{proof}
    See Appendix \ref{app:MV}.
\end{proof}
\begin{remark}
    Since mean-variance hedging is performed under the risk neutral measure $P$, the initial capital $\pi_0$ for the mean-variance strategy coincides with the RS price \eqref{equ:rs}.
\end{remark}

\section{Hedges on simulated data}\label{sec:simulation}
To investigate the model and hedging strategies, we run a simulation experiment on the regime switching model.
We apply different pricing formulas and different hedging strategies with a specified set of parameters.
To compare with the empirical study of the last section, we use a realistic set of parameters corresponding to HKDUSD.
Therefore, throughout this section, we assume market data
\begin{equation*}
    S_0 = 7.8, \quad K = 7.8, \quad T = 0.5, \quad r_d = 1\%,\quad \text{and} \quad r_f = 1.5\%   
\end{equation*}
and RS model parameters
\begin{equation*}
    \theta = (\ubar{\sigma}, \bar{\sigma}, \lambda, u, \delta) = (0.5\%, 10\%, 0.2, -0.01, 0)
\end{equation*}
We also consider the performance of the hedging strategies conditioned on both scenarios, that is, in the presence or absence of jump.
The hedging error which is defined as
\begin{equation}\label{equ:hedge_err}
   \text{Error}_{hedge} = \left|\frac{\text{Portfolio}(T) - C(T)}{K}\right| \times 100\%
\end{equation}
which measures the performance of hedging at the expiry time $T$.
As seen in the case of quadratic hedging thereafter, we also consider a mean tracking error over the whole hedging period, given by
\begin{equation}\label{equ:hedge_mte}
    \text{MTError}_{hedge} = \frac{1}{n}\sum_{i=1}^n \left|\frac{\text{Portfolio}(t_i) - C(t_i)}{K}\right| \times 100\%
\end{equation}
where $C(t_i)$ represents the RS price at time $t_i = i \ dt$, for $i=1,2,\cdots,130$, and
\begin{equation*}
    \text{Portfolio}(t_i) = \eta^1(t_{i-1}) S_{t_i} + \eta^0(t_{i-1}) e^{(r_d-r_f)/130}  
\end{equation*}
where $\eta^0(\cdot)$ represents the units of the domestic bond and $\eta^1(\cdot)$ denotes the amount invested in the underlying asset.

As for the different hedging strategies, we consider the following five types:
\begin{itemize}
    \item \textbf{BS Delta}: Apply the BS Delta given in \eqref{equ:gkdelta}.
    	This naive strategy takes the lower level volatility parameter $\ubar\sigma$ to be the BS volatility parameter before the regime switching and the upper level volatility parameter $\bar\sigma$ after the regime switching.
    \item \textbf{RS Delta}: Apply the RS Delta given in \eqref{equ:rsdelta}.
    \item \textbf{Approximated RS Delta}: Apply the approximated RS Delta given in \eqref{equ:approx_RSdelta}.
    \item \textbf{Mean-variance with RS Delta}: Apply the Mean-variance hedging portfolio given in \eqref{equ:mvdelta} with RS delta.\footnote{That is, for $\partial \tilde{\mathcal{C}}/\partial S$ in \eqref{equ:approx_RSdelta}.}
    \item \textbf{Mean-var\-iance with approximated RS Delta}: Apply the Mean-va\-riance hedging portfolio given in \eqref{equ:mvdelta} with approximated RS delta.
\end{itemize}
For the last four hedging strategies, we use the RS model parameters $\theta$.

\subsection{Simulation result: No jump}
The results of hedging errors on 10,000 simulation paths with no regime switching scenario are listed in Table \ref{tab:simu_nojump_err}.
\begin{table}[H]
    \begin{center}
        \resizebox{\columnwidth}{!}{
            \begin{tabular}{@{}llrrrrrrr@{}}
                \toprule
                \multicolumn{1}{c}{} &\multicolumn{1}{c}{} & \multicolumn{1}{c}{Mean} & \multicolumn{1}{c}{Std} & \multicolumn{1}{c}{Min}& \multicolumn{1}{c}{q-25\%}& \multicolumn{1}{c}{Median}& \multicolumn{1}{c}{q-75\%}& \multicolumn{1}{c}{Max}\\
                \midrule
                             &BS&0.149\%    &0.009    &0.080\%    &0.144\%    &0.149\%    &0.154\%    &0.209\%\\
                             &RS&0.154\%    &0.022    &0.098\%    &0.135\%    &0.151\%    &0.173\%    &0.228\%\\
$\text{Error}_{hedge}$       &Approx RS&0.154\%    &0.023    &0.098\%    &0.135\%    &0.151\%    &0.174\%    &0.229\%\\
\cmidrule{3-9}
                             &MV&0.216\%    &0.152    &0.000\%    &0.093\%    &0.190\%    &0.311\%    &1.090\%\\
                             &Approx MV&0.212\%    &0.148    &0.000\%    &0.093\%    &0.188\%    &0.305\%    &1.059\%\\
\midrule
                             &BS&0.089\%    &0.009    &0.058\%    &0.083\%    &0.091\%    &0.097\%    &0.107\%\\
                             &RS&0.093\%    &0.009    &0.066\%    &0.086\%    &0.092\%    &0.100\%    &0.118\%\\
$\text{MTError}_{hedge}$     &Approx RS&0.093\%    &0.009    &0.063\%    &0.085\%    &0.092\%    &0.100\%    &0.120\%\\
\cmidrule{3-9}
                             &MV&0.167\%    &0.098    &0.026\%    &0.090\%    &0.143\%    &0.222\%    &0.812\%\\
                             &Approx MV&0.162\%    &0.095    &0.026\%    &0.088\%    &0.139\%    &0.216\%    &0.787\% \\              
                \bottomrule
            \end{tabular}
        }
        \caption{Hedging errors and RMSEs on 10,000 simulated data without jump: BS delta, RS delta, approximated RS delta, Mean-variance and approximated Mean-variance hedging.}
        \label{tab:simu_nojump_err}
    \end{center}
\end{table}
Figure \ref{fig:sim_no_heging_result} and Figure \ref{fig:sim_no_heging_result_rmse} illustrate the distribution of hedging errors for these five hedging strategies.

\begin{figure}[H]
    \centering
    \begin{minipage}{0.495\textwidth}
    \includegraphics[width=\textwidth]{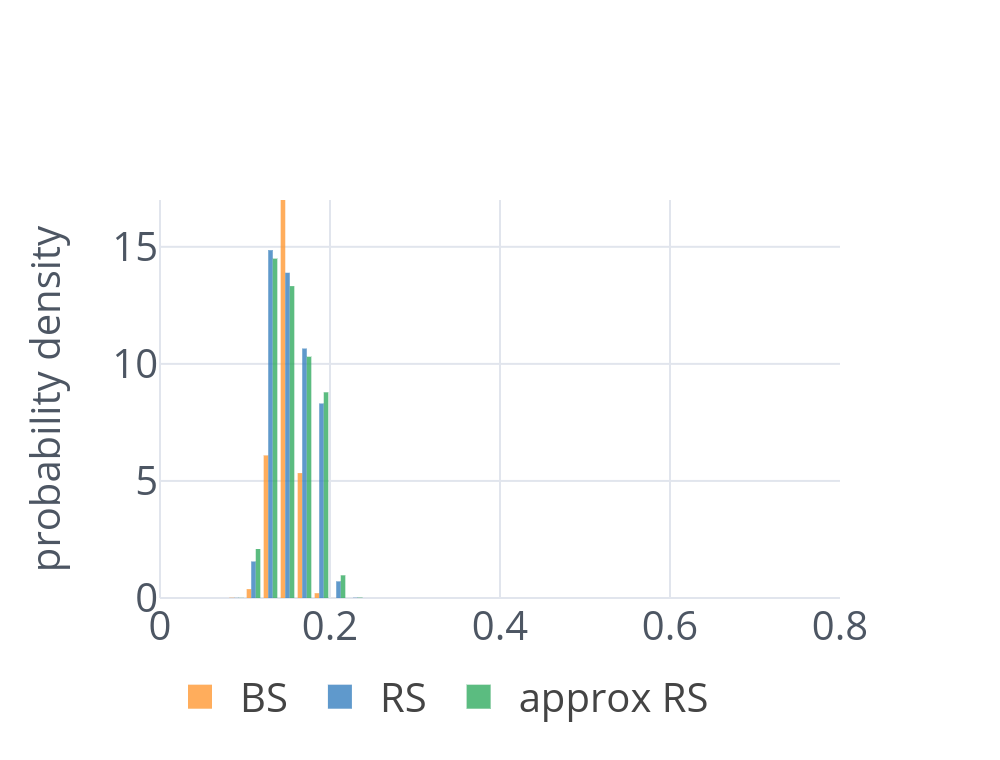}
    \end{minipage}
    \begin{minipage}{0.495\textwidth}
    \includegraphics[width=\textwidth]{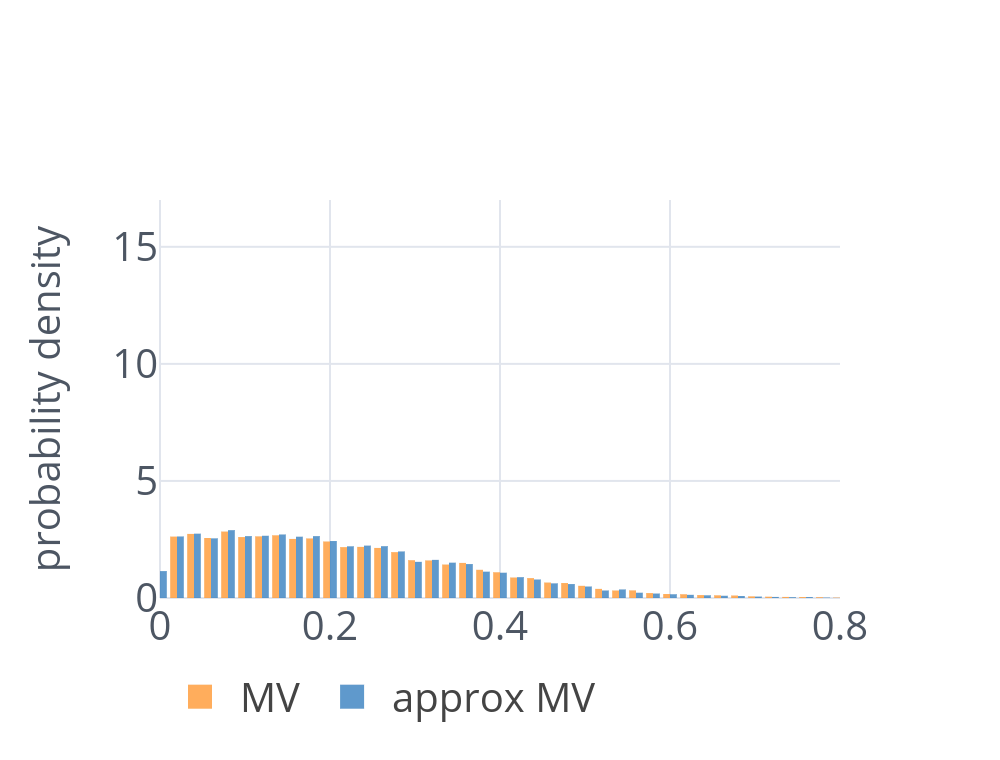}
    \end{minipage}
    \caption{Hedging errors on 10,000 simulated data without jump: BS delta, RS delta, approximated RS delta, Mean-variance and approximated Mean-variance hedging.}
    \label{fig:sim_no_heging_result}
\end{figure}

\begin{figure}[H]
    \centering
    \begin{minipage}{0.495\textwidth}
    \includegraphics[width=\textwidth]{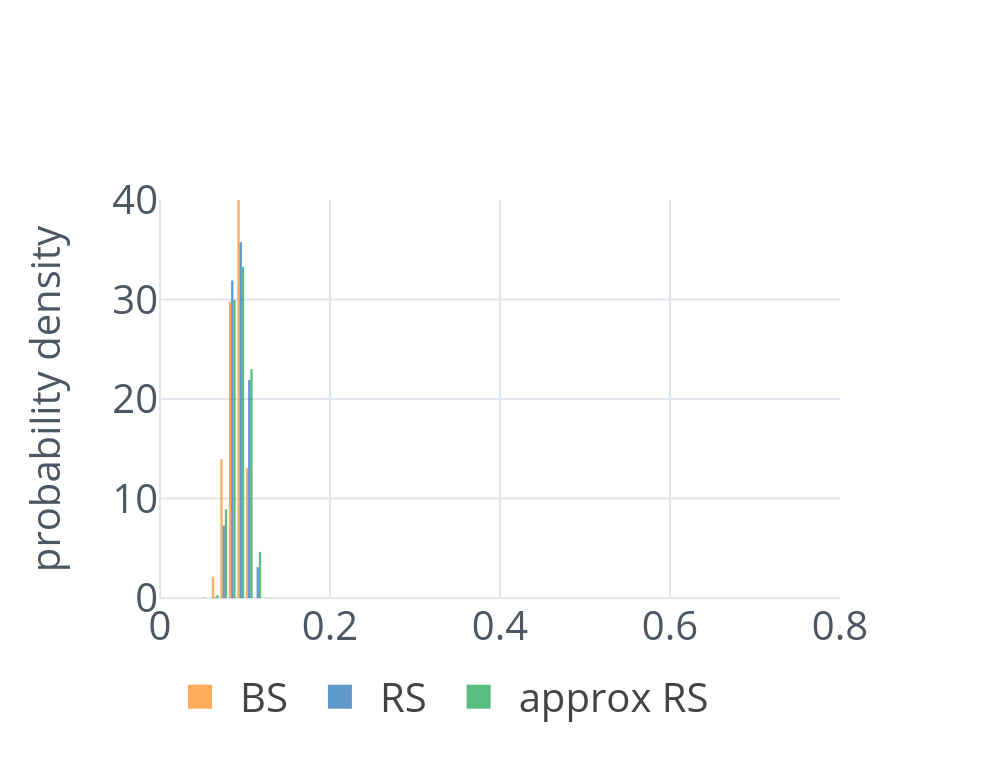}
    \end{minipage}
    \begin{minipage}{0.495\textwidth}
    \includegraphics[width=\textwidth]{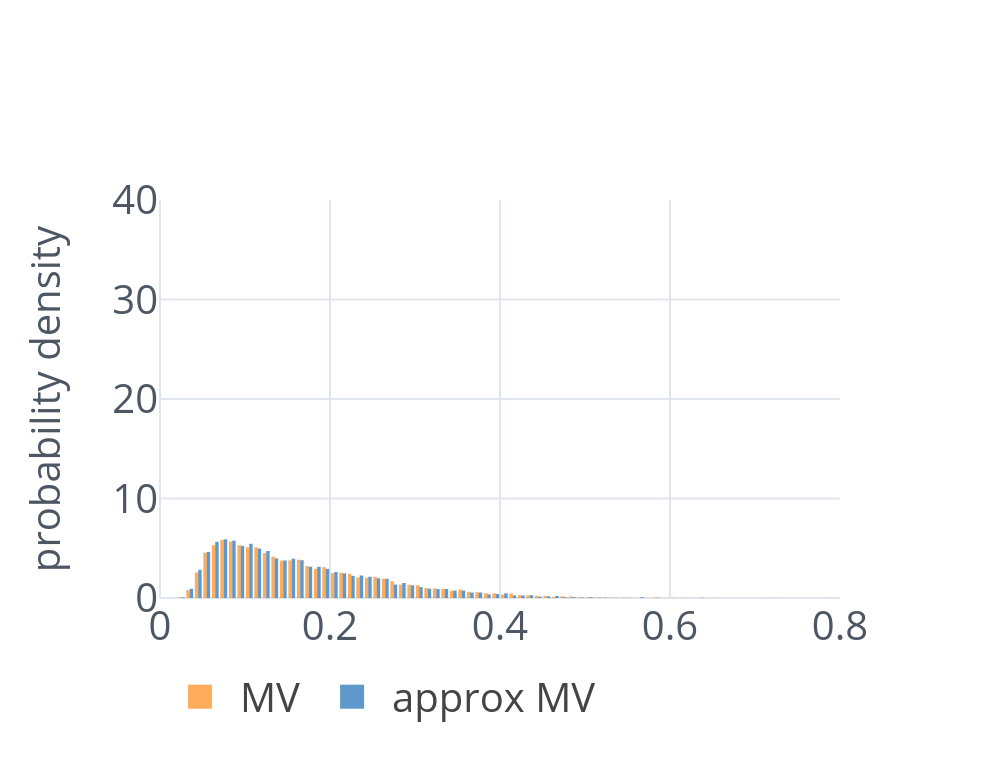}
    \end{minipage}
    \caption{Hedging RMSEs on 10,000 simulated data without jump: BS delta, RS delta, approximated RS delta, Mean-variance and approximated Mean-variance hedging.}
    \label{fig:sim_no_heging_result_rmse}
\end{figure}

\subsection{Simulation results: With jump}
The results of hedging errors on 10,000 simulation paths with no regime switching scenario are listed in Table \ref{tab:simu_jump_err}.
\begin{table}[H]
    \begin{center}
        \resizebox{\columnwidth}{!}{
            \begin{tabular}{@{}llrrrrrrr@{}}
                \toprule
                \multicolumn{1}{c}{} &\multicolumn{1}{c}{} & \multicolumn{1}{c}{Mean} & \multicolumn{1}{c}{Std} & \multicolumn{1}{c}{Min}& \multicolumn{1}{c}{q-25\%}& \multicolumn{1}{c}{Median}& \multicolumn{1}{c}{q-75\%}& \multicolumn{1}{c}{Max}\\
                \midrule
                                &BS          &1.460\%    &0.669    &0.000\%    &0.984\%    &1.565\%    &1.992\%    &3.186\%\\
                                &RS          &1.508\%    &0.687    &0.000\%    &1.016\%    &1.618\%    &2.059\%    &3.297\%\\
$\text{Error}_{hedge}$          &Approx RS    &1.511\%    &0.688    &0.000\%    &1.018\%    &1.620\%    &2.062\%    &3.301\%\\
\cmidrule{3-9}
                                &MV          &0.911\%    &0.628    &0.000\%    &0.389\%    &0.829\%    &1.328\%    &3.954\%\\
                                &Approx MV    &0.921\%    &0.628    &0.000\%    &0.398\%    &0.845\%    &1.339\%    &3.914\%\\
\midrule
                                &BS          &0.937\%    &0.670    &0.068\%    &0.315\%    &0.832\%    &1.480\%    &2.764\%\\
                                &RS          &0.972\%    &0.697    &0.063\%    &0.325\%    &0.858\%    &1.546\%    &2.748\%\\
$\text{MTError}_{hedge}$        &Approx RS   &0.974\%    &0.698    &0.060\%    &0.325\%    &0.861\%    &1.548\%    &2.753\%\\
\cmidrule{3-9}
                                &MV          &0.612\%    &0.507    &0.021\%    &0.221\%    &0.444\%    &0.882\%    &3.705\%\\
                                &Approx MV   &0.619\%    &0.510    &0.021\%    &0.221\%    &0.451\%    &0.896\%    &3.668\%\\

                \bottomrule
            \end{tabular}
        }
        \caption{Hedging errors and RMSEs on 10,000 simulated data with jump: BS delta, RS delta, approximated RS delta, Mean-variance and approximated Mean-variance hedging.}
        \label{tab:simu_jump_err}
    \end{center}
\end{table}
Figure \ref{fig:sim_heging_result} and Figure \ref{fig:sim_heging_result_rmse} illustrate the distribution of hedging errors for these five hedging strategies.

\begin{figure}[H]
    \centering
    \begin{minipage}{0.495\textwidth}
    \includegraphics[width=\textwidth]{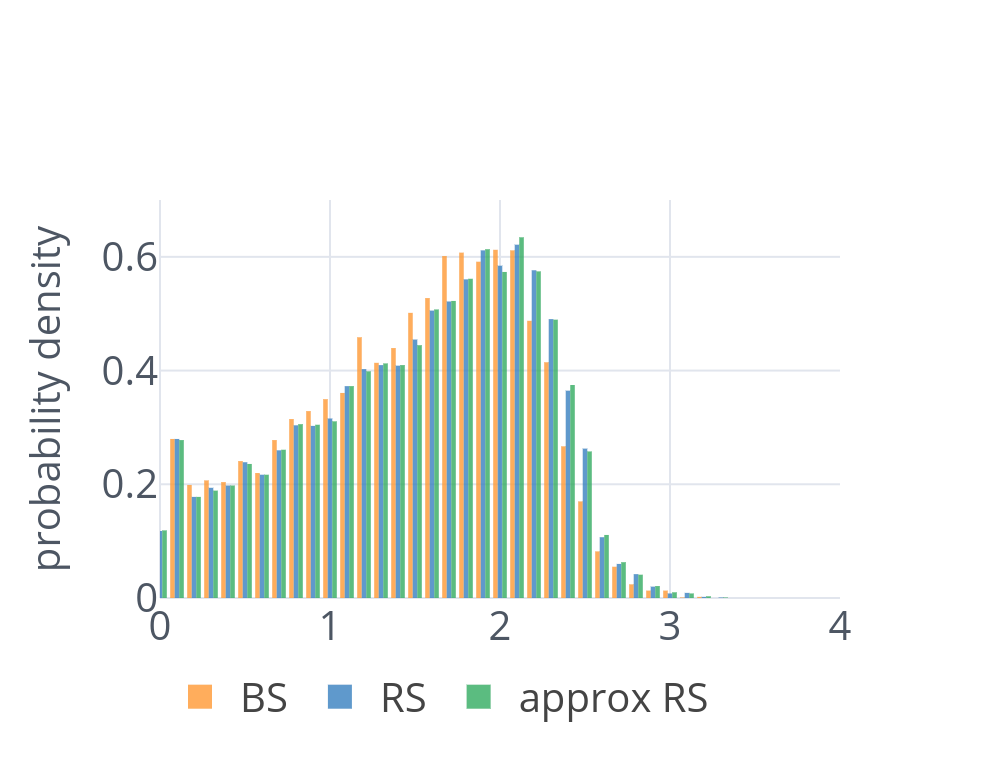}
    \end{minipage}
    \begin{minipage}{0.495\textwidth}
    \includegraphics[width=\textwidth]{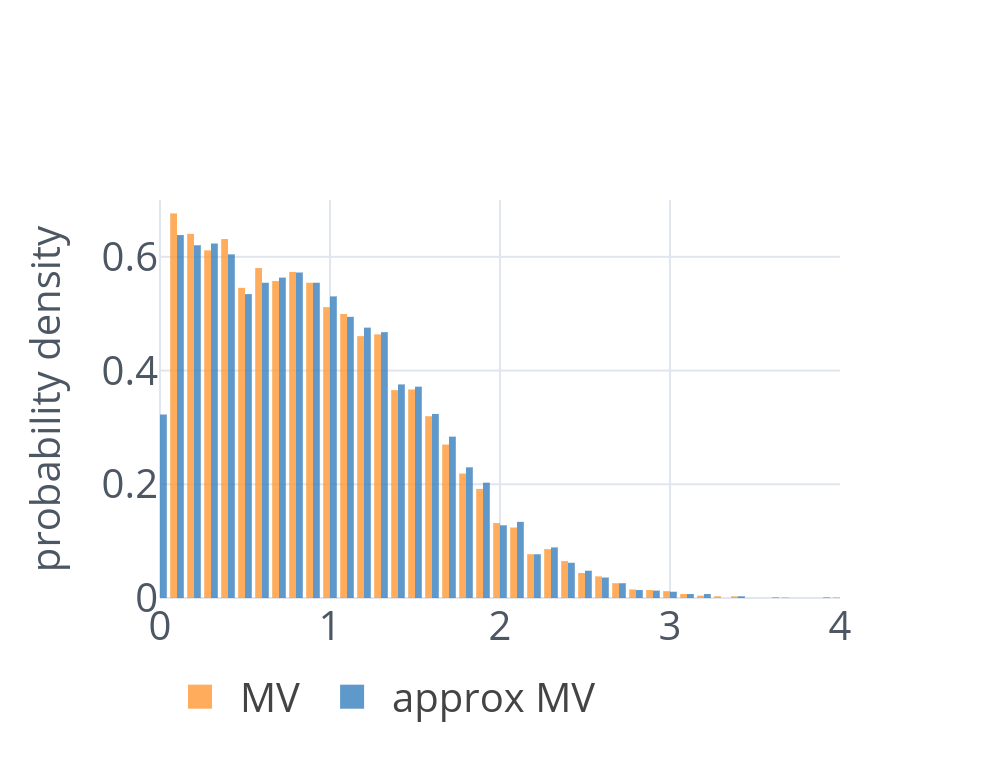}
    \end{minipage}
    \caption{Hedging errors on 10,000 simulated data with jump: BS delta, RS delta, approximated RS delta, Mean-variance and approximated Mean-variance hedging.}
    \label{fig:sim_heging_result}
\end{figure}

\begin{figure}[H]
    \centering
    \begin{minipage}{0.495\textwidth}
    \includegraphics[width=\textwidth]{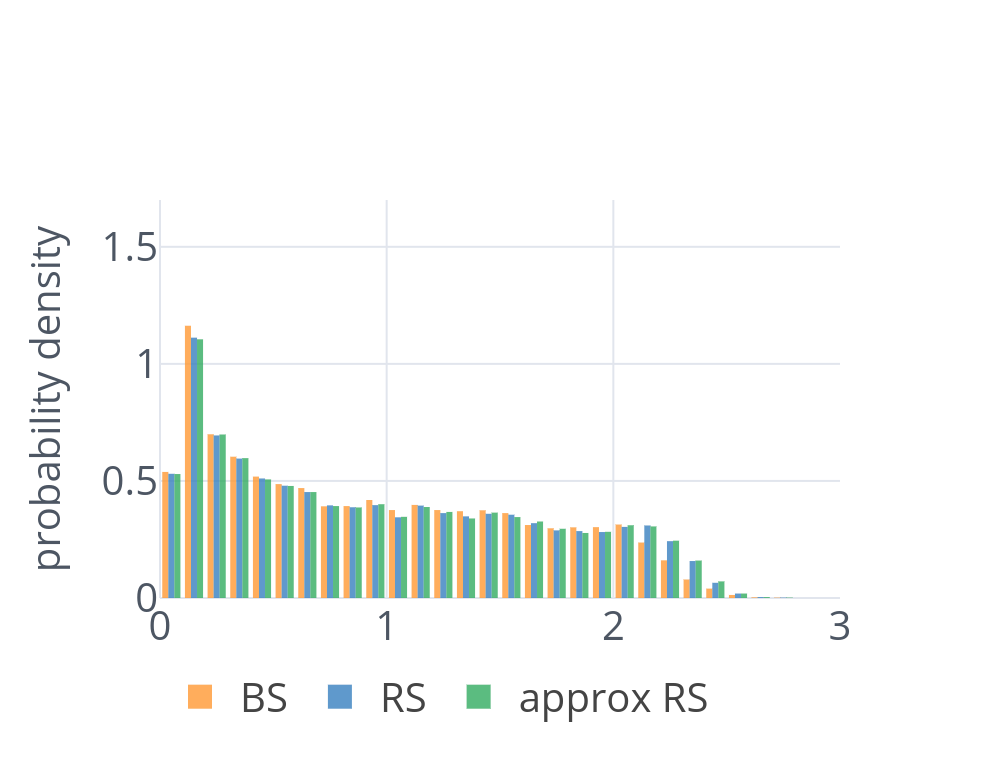}
    \end{minipage}
    \begin{minipage}{0.495\textwidth}
    \includegraphics[width=\textwidth]{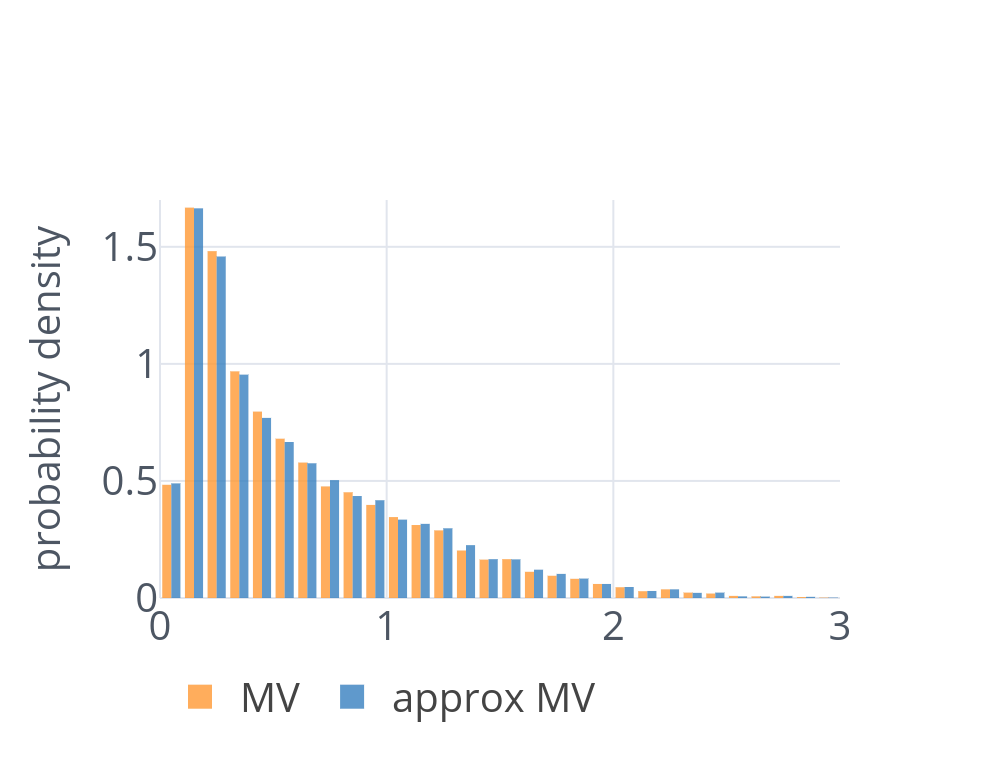}
    \end{minipage}
    \caption{Hedging RMSEs on 10,000 simulated data with jump: BS delta, RS delta, approximated RS delta, Mean-variance and approximated Mean-variance hedging.}
    \label{fig:sim_heging_result_rmse}
\end{figure}

\subsection{Computational time}
All the computations were performed using Python on an Intel(R) Core(TM) i7-6500U CPU @ 2.50GHz PC with 8GB of RAM without particular optimization scheme.
Computational time of these five different delta-hedging is listed in the Table \ref{tab:delta_time}.
Note that the computational time is tested on one single path without jump, which involves 130 times computation for deltas on each path.
\begin{table}[H]
    \begin{center}
        \resizebox{\columnwidth}{!}{
            \begin{tabular}{@{}llrrrr@{}}
                \toprule
                \multicolumn{1}{c}{} &\multicolumn{1}{c}{BS} & \multicolumn{1}{c}{RS} & \multicolumn{1}{c}{Approx RS} & \multicolumn{1}{c}{Mean-variance}& \multicolumn{1}{c}{Approx Mean-variance}\\
                \midrule
                Computational time on Deltas    &0.0095 sec  &5.1300 sec  &0.0551 sec &12.600 sec &0.176 sec  \\
                \bottomrule
            \end{tabular}
        }
        \caption{computational time.}
        \label{tab:delta_time}
    \end{center}
\end{table}
The run times show that the approximated RS delta hedge is significantly faster than the exact one by a factor of about 93.
Moreover, the mean-variance hedge with approximated RS is also significantly faster than the mean-variance with RS by a factor of about 71.

\section{Calibration performance}\label{sec:calibration}
In this section, we present the empirical performance in terms of accuracy and computational time of the calibration to market quotations based on the regime switching model.

\subsection{Dataset}
We consider the HKDUSD foreign exchange option daily data quotations for the period 2014-01-01 to 2019-01-10, downloaded from Bloomberg.
We collect daily quotations for each maturity: 1D, 1W, 1M, 3M, 6M, and 1Y.
Each maturity has five quotations: $\sigma_{10RR}$, $\sigma_{25RR}$, $\sigma_{10BF}$, $\sigma_{25BF}$, and $\sigma_{ATM}$.
Also, for each of the foreign and domestic, we collect the interest rate for each of these maturities.
We define the maturity $T_{1Y} = 1$ and subsequently, $T_{6M} = 1/2$, $T_{3M} = 1/4$, $T_{1M} = 1/12$, $T_{1W} = 1/52$ and $T_{1D} = 1/260$.
Table \ref{tab:hkdusd_data}, shows a snapshot of the available data for maturity 1M.
\begin{table}[H]
    \begin{center}
        \resizebox{\columnwidth}{!}{
            \begin{tabular}{@{}ccccccccccc@{}}
                \toprule
                \multicolumn{1}{c}{Date} & \multicolumn{1}{c}{S} & \multicolumn{1}{c}{Rd} & \multicolumn{1}{c}{Rf}& \multicolumn{1}{c}{ATM}& \multicolumn{1}{c}{25RR}& \multicolumn{1}{c}{25BF}& \multicolumn{1}{c}{10RR}& \multicolumn{1}{c}{10BF}\\

                \midrule
                2014-01-01 & 7.75407 & 0.005488 & 0.003482 & 0.005000 & -0.004250 & 0.003125 & -0.007750 & 0.012542\\
                2014-01-02 & 7.75398 & 0.005475 & 0.003475 & 0.005000 & -0.004250 & 0.003117 & -0.007750 & 0.012542\\
                2014-01-03 & 7.75405 & 0.005464 & 0.003465 & 0.005000 & -0.004250 & 0.003108 & -0.007583 & 0.012542\\
                \vdots     & \vdots  & \vdots   & \vdots   & \vdots   & \vdots    & \vdots   & \vdots    & \vdots\\
                2019-01-09 & 7.83774 & 0.023392 & 0.028570 & 0.008750 & -0.004967 & 0.001550 & -0.009533 & 0.006167\\
                2019-01-10 & 7.83847 & 0.023219 & 0.028609 & 0.008717 & -0.004917 & 0.001550 & -0.009317 & 0.006117\\
                \bottomrule
            \end{tabular}
        }
        \caption{HKDUSD FX data with maturity 1M provided by Bloomberg}
        \label{tab:hkdusd_data}
    \end{center}
\end{table}
Concerning the conventions of quotations on the FX market, we refer to \citet{clark2011} for more details.
We adopt the following notation: $pa$, for premium adjusted valuation, as well as $s$ and $f$ for the spot and forward valuation, respectively.
Bloomberg's conventions of HKDUSD are given in Table \ref{tab:hkdusd_quot}.
\begin{table}[H]
    \begin{center}
        \begin{tabular}{@{}lcccccccccc@{}}
            \toprule
            \multicolumn{1}{c}{} & \multicolumn{1}{c}{1D} & \multicolumn{1}{c}{1W} & \multicolumn{1}{c}{1M}& \multicolumn{1}{c}{3M}& \multicolumn{1}{c}{6M}& \multicolumn{1}{c}{1Y}& \multicolumn{1}{c}{2Y}\\

            \midrule
            convention       & s-pa&	s-pa&	s-pa&	s-pa&	s-pa&	f-pa&	f-pa\\
            \bottomrule
        \end{tabular}
        \caption{HKDUSD conventions in Bloomberg per maturity.}
        \label{tab:hkdusd_quot}
    \end{center}
\end{table}
Also, according to the conventions from Bloomberg's dataset, $\sigma_{10RR}$ and $\sigma_{10BF}$ (as well as $\sigma_{25RR}$ and $\sigma_{25BF}$) are related to delta-volatility pairs via
\begin{align*}
    \sigma_{10C} &= \sigma_{ATM} + \sigma_{10BF} + \frac{\sigma_{10RR}}{2} &
    \sigma_{10P} &= \sigma_{ATM} + \sigma_{10BF} - \frac{\sigma_{10RR}}{2}.	
\end{align*}
The strikes $K_{ATM}$, $K_{10P}$,$ K_{10C}$, $K_{25P}$, and $K_{25C}$ can then be recovered according to the conventions, see \citep{clark2011}.
For each maturity, we therefore have five strike-volatility pairs,
\begin{equation*}
(K_{10P}, \sigma_{10P}), \quad (K_{25P}, \sigma_{25P}), \quad (K_{ATM}, \sigma_{ATM}), \quad (K_{25C}, \sigma_{25C}), \quad (K_{10C}, \sigma_{10C})
\end{equation*}
as a basis for the model calibration.

\subsection{Single Maturity Calibration}
Given a generic model $\sigma_{model}(K,\theta)$ for the volatility smile, the calibration search for an optimal parameter $\theta^\ast$ that minimize the sum of squared differences between observed and model-based implied volatilities using classical least-square method.
We measure the performance of the calibration by mean error and root mean square error
\begin{align}
    \text{ME} & = \frac{1}{5}\sum_{i\in \mathcal{I}}\left|\frac{\sigma_{model}(K_{i};\theta^*) - \sigma_{i}}{\sigma_{i}}\right|\times100\%\\
    \text{RMSE} & = \sqrt{\frac{1}{5}\sum_{i\in \mathcal{I}}\left(\frac{\sigma_{model}(K_{i};\theta^*) - \sigma_{i}}{\sigma_{i}}\right)^2}\times100\%\label{equ:calibration_err}
\end{align}
where $\sigma_{model}(K_{i};\theta^*)$ and $\sigma_{i}$ denote the model-based implied volatilities and market quotations for $i\in \mathcal{I} = \{10P,25P, ATM, 25C, 10C\}$.

In the following we compare the performance of the implied volatility generated by 
\begin{itemize}
    \item SABR model\footnote{By \citet{hagan2002}, see Appendix \ref{app:SABR}.} (Benchmark);
    \item RS model;
    \item Approximated RS model;
\end{itemize}
Table \ref{tab:calibration_me} and Table \ref{tab:calibration_rmse} show the mean errors and root mean square errors of calibrations for different maturities.
Across all maturities, the RS calibration gives the smallest error in the mean as well as the median sense.
Note that, in these two table, the errors are skewed as the median and 75\% quantile are showing with respect to the mean.
Hence, up to a couple of outliers, RS calibration performs very good.
Compared with the SABR calibration benchmark, the error is generally smaller by a factor of 2 to 4 in the median sense.
As for the approximated RS calibration, the accuracy is particularly poor, around 15\% since approximated RS price is usually higher than the RS one.

\begin{table}[H]
    \begin{center}
        \resizebox{\columnwidth}{!}{
            \begin{tabular}{@{}llrrrrrrr@{}}
                \toprule
                \multicolumn{1}{c}{} &\multicolumn{1}{c}{} & \multicolumn{1}{c}{Mean} & \multicolumn{1}{c}{Std} & \multicolumn{1}{c}{Min}& \multicolumn{1}{c}{q-25\%}& \multicolumn{1}{c}{Median}& \multicolumn{1}{c}{q-75\%}& \multicolumn{1}{c}{Max}\\
                \midrule
                                      & SABR      & 1.58\%          & 0.73 & 0.02\%  & 1.21\%  & 1.01\%          & 2.12\%  & 3.40\%\\
                1M                    & RS        & \textbf{0.77\%} & 0.59 & 0.00\%  & 0.29\%  & \textbf{0.66\%} & 1.20\%  & 4.62\%\\
                                      & Approx RS & 15.27\%         & 2.74 & 7.69\%  & 12.89\% & 16.03\%         & 17.81\% & 18.94\%\\
                    \midrule
                					  & SABR      & 1.96\%          & 0.79 & 0.34\%  & 1.46\%  & 2.10\%          & 2.50\%  & 6.46\%\\
                3M                    & RS        & \textbf{0.88\%} & 0.70 & 0.00\%  & 0.44\%  & \textbf{0.71\%} & 1.18\%  & 6.18\%\\
                                      & Approx RS & 15.62\%         & 2.11 & 10.39\% & 13.66\% & 15.96\%         & 17.58\% & 18.79\%\\
                    \midrule
                                      & SABR      & 2.04\%          & 0.63 & 0.46\%  & 1.51\%  & 2.18\%          & 2.53\%  & 3.74\%\\
                6M                    & RS        		& \textbf{0.75\%} & 0.74 & 0.00\%  & 0.30\%  & \textbf{0.54\%} & 0.90\%  & 4.57\%\\
                                      & Approx RS & 15.60\%         & 1.74 & 11.59\% & 13.94\% & 15.68\%         & 16.92\% & 18.67\%\\
                    \midrule
                                      & SABR      & 2.13\%          & 0.71 & 0.40\%  & 1.53\%  & 2.32\%          & 2.64\%  & 4.10\%\\
                1Y                    & RS        & \textbf{0.90\%} & 0.96 & 0.00\%  & 0.31\%  & \textbf{0.55\%} & 1.13\%  & 4.90\%\\
                                      & Approx RS & 15.68\%         & 1.50 & 11.54\% & 14.82\% & 15.62\%         & 16.92\% & 18.56\%\\
                    \bottomrule
            \end{tabular}
        }
        \caption{Mean errors for calibration}
        \label{tab:calibration_me}
    \end{center}
\end{table}

\begin{table}[H]
    \begin{center}
        \resizebox{\columnwidth}{!}{
            \begin{tabular}{@{}llrrrrrrr@{}}
                \toprule
                \multicolumn{1}{c}{} &\multicolumn{1}{c}{} & \multicolumn{1}{c}{Mean} & \multicolumn{1}{c}{Std} & \multicolumn{1}{c}{Min}& \multicolumn{1}{c}{q-25\%}& \multicolumn{1}{c}{Median}& \multicolumn{1}{c}{q-75\%}& \multicolumn{1}{c}{Max}\\
                \midrule
                                      & SABR      & 1.92\%          & 0.89 & 0.02\%  & 1.21\%  & 1.97\%          & 2.58\%  & 4.34\%\\
                  1M                  & RS        & \textbf{0.99\%} & 0.78 & 0.00\%  & 0.35\%  & \textbf{0.82\%} & 1.53\%  & 6.13\%\\
                                      & Approx RS & 15.52 \%        & 2.69 & 7.80\%  & 13.29\% & 16.29\%         & 17.97\% & 19.08\%\\
                    \midrule
                                      & SABR      & 2.41\%          & 1.01 & 0.41\%  & 1.69\%  & 2.60\%          & 3.14\%  & 7.55\%\\
                  3M                  & RS        & \textbf{1.14\%} & 0.92 & 0.00\%  & 0.55\%  & \textbf{0.92\%} & 1.56\%  & 7.78\%\\
                                      & Approx RS & 15.94\%         & 2.10 & 10.64\% & 14.07\% & 16.49\%         & 17.81\% & 19.11\%\\
                    \midrule
                                      & SABR      & 2.46\%          & 0.78 & 0.61\%  & 1.81\%  & 2.62\%          & 3.07\%  & 4.32\%\\
                  6M                  & RS        & \textbf{0.98\%} & 0.98 & 0.00\%  & 0.38\%  & \textbf{0.71\%} & 1.16\%  & 6.03\%\\
                                      & Approx RS & 15.94\%         & 1.71 & 11.92\% & 14.31\% & 16.09\%         & 17.30\% & 18.95\%\\
                    \midrule
                                      & SABR      & 2.57\%          & 0.87 & 0.53\%  & 1.95\%  & 2.75\%          & 3.16\%  & 4.99\%\\
                  1Y                  & RS        & \textbf{1.19\%} & 1.26 & 0.00\%  & 0.41\%  & \textbf{0.71\%} & 1.45\%  & 6.48\%\\
                                      & Approx RS & 16.07\%         & 1.40 & 12.76\% & 15.03\% & 16.03\%         & 17.26\% & 18.88\%\\
                    \bottomrule
            \end{tabular}
        }
        \caption{Root mean square errors for calibration}
        \label{tab:calibration_rmse}
    \end{center}
\end{table}
Figure \ref{fig:three_smile_dash} shows on a sample day, the resulting smiles of the three calibration methods together with the quotations.
The volatility smile given by RS calibration procedure do not only fits the market quotations very well but also shows the "bird" shape of the smile.
\begin{figure}[H]
    \centering
    \includegraphics[width=\textwidth]{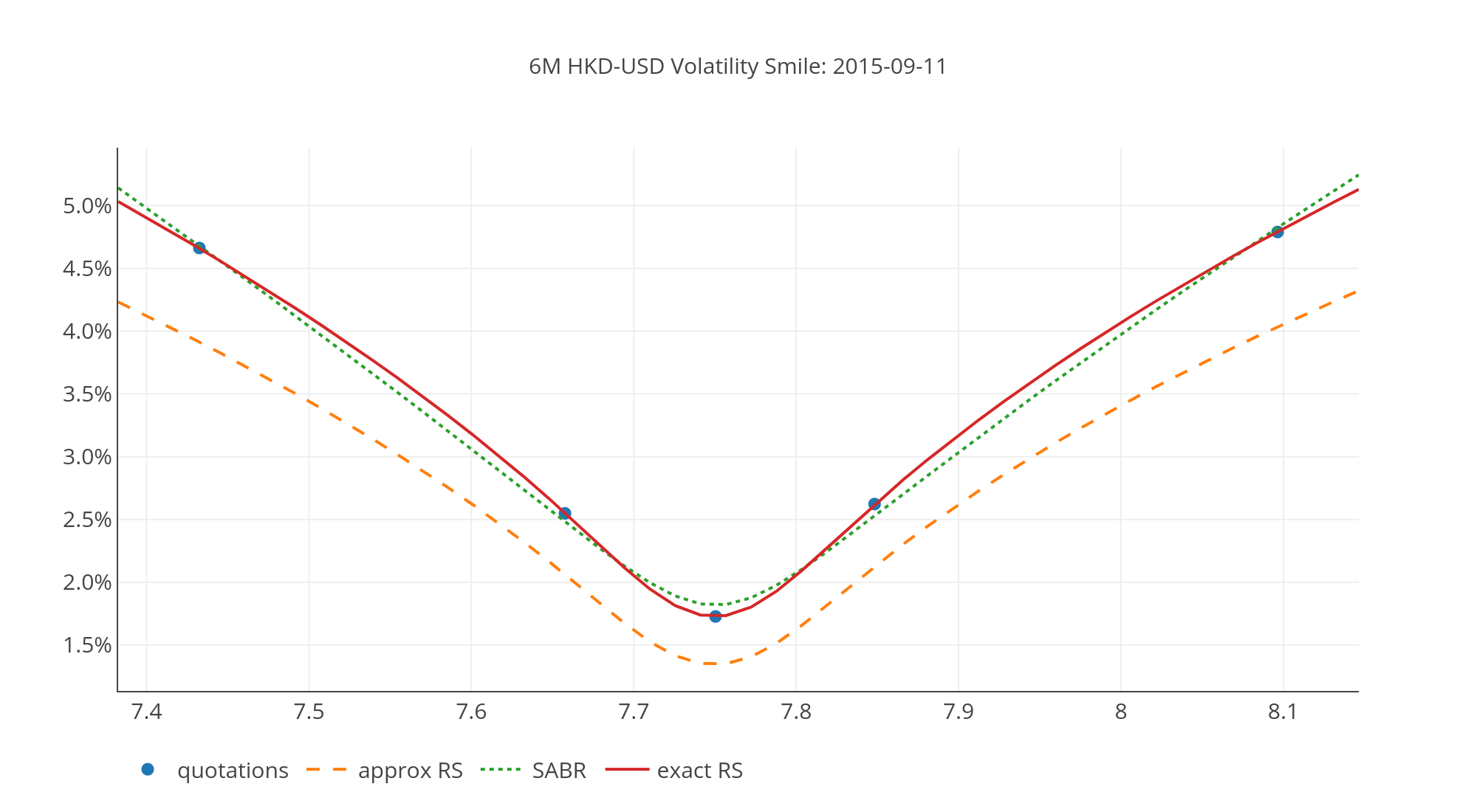}
    \caption{Calibration result: SABR, exact RS solution and approximated RS solution.}
    \label{fig:three_smile_dash}
\end{figure}

Even if the approximated RS calibration is computationally faster, the accuracy is not satisfactory as the error is constantly large.
We therefore exclude approximated RS pricer for parameter calibration.

Figure \ref{fig:calibration_error} shows that the mean errors of the calibrated implied volatilities for the SABR and RS parametrization for maturities at 6M and 1Y over time.
The RS model outperforms the SABR one almost constantly over time.
\begin{figure}[H]
    \centering
    \includegraphics[width=\textwidth]{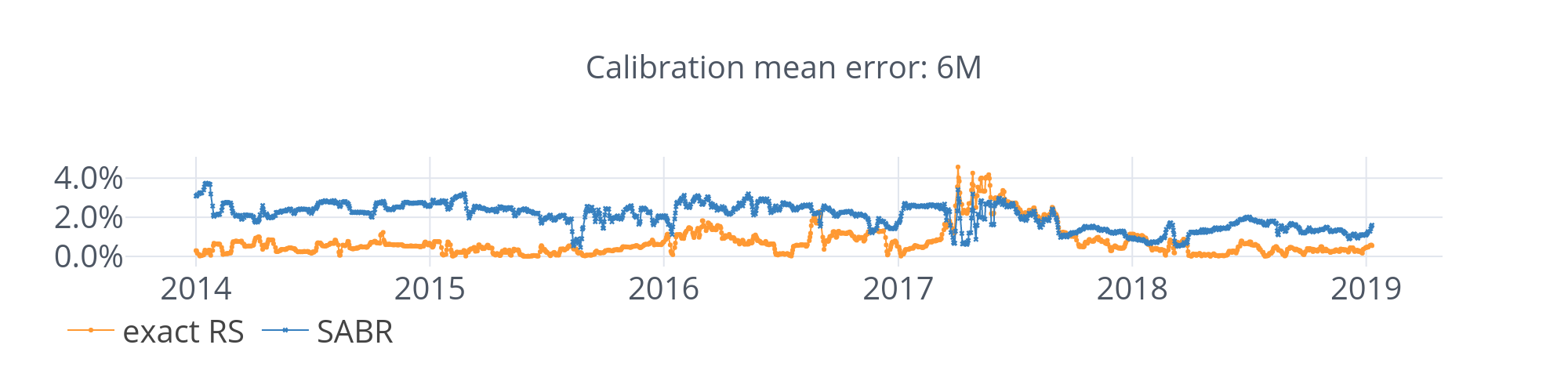}
\end{figure}
\begin{figure}[H]
    \centering
    \includegraphics[width=\textwidth]{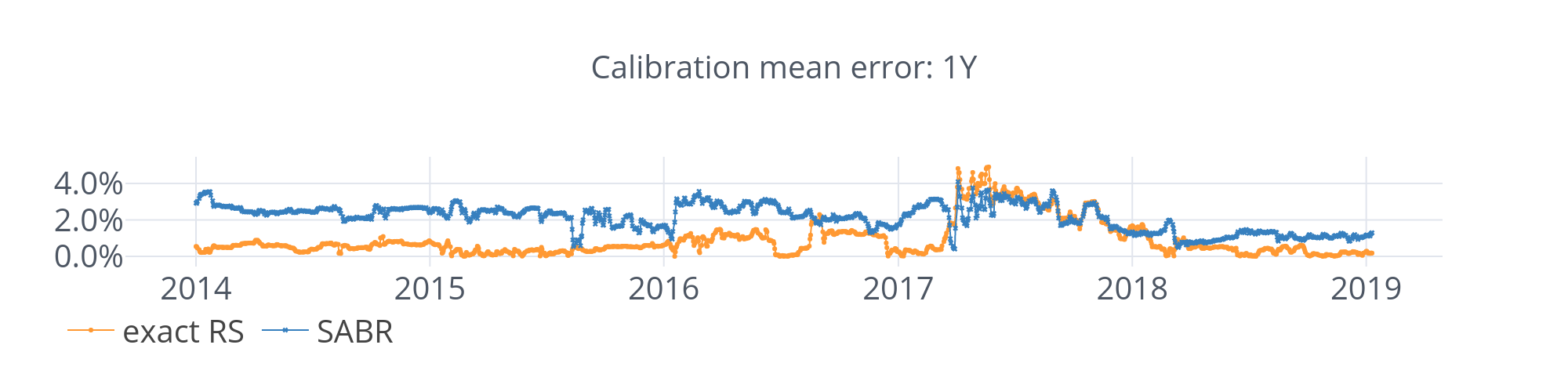}
    \caption{Comparison of the calibration mean errors: SABR and RS. }
    \label{fig:calibration_error}
\end{figure}
Furthermore, in a SABR framework, the volatility dynamic $\sigma(T)$ at given time $T$ follows a log-normal distribution $Log\mathcal{N}\left(\ln(a)-\frac{1}{2}b^2T, b^2T\right)$.
Figure \ref{SABR_vol_fig} shows a consequent discrepancy between the SABR theoretical\footnote{From the SABR calibrated parameters where $a$ and $b$ are taken in the mean sense.} and the empirical distribution of the volatility.
Hence aside accuracy issue, SABR model does not seem to be an appropriate choice as a model for the pegged market.
\begin{figure}[H]
    \centering
    \includegraphics[width=0.85\textwidth]{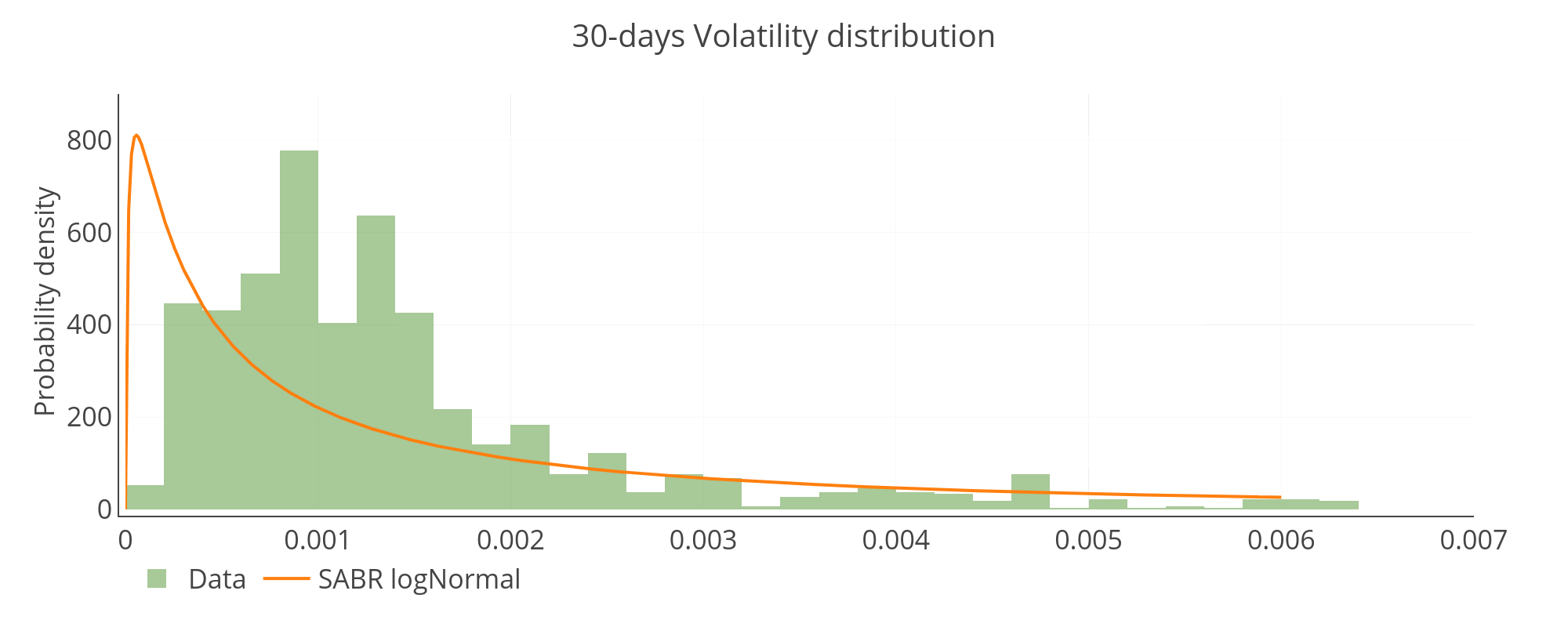}
    \caption{Theoretical SABR versus 30 days empirical volatility distribution}
    \label{SABR_vol_fig}
\end{figure}
Thus, for the reasons provided there above, we rule out SABR model as well as approximated RS model as a calibration approach in what follows.

\subsection{Surface calibration}
For hedging purposes, an interpolation of the parameters between the given tenors is necessary to obtain the corresponding implied volatility surface.
We proceed as follows:
At date $t_0$, we are given quotations $\sigma_{i}(t_0,T_{\zeta})$ for $i \in \mathcal{I}=\{10P, 25P, ATM, 25C, 10C\}$ and $\zeta\in\mathcal{T} = \{1D, 1W, 1M, 3M, 6M\}$.
For $T_{3M} \leq t \leq T_{6M}$, we set the quotations $\sigma_i(t_0,t)$ for maturity $t$ by interpolation according to
\begin{equation}\label{equ:interp_quots}
    \sigma_i(t_0,t) = \sqrt{\frac{T_{6M}\left( t-T_{3M} \right)\sigma^2_i(t_0, T_{6M}) +T_{3M}\left( T_{6M}-t \right)\sigma_i^2(t_0, T_{3M})}{t\left( T_{6M} - T_{3M} \right)}}
\end{equation}
for any $i$ in $\mathcal{I}$.
The same procedure is applied for $T_{1M}\leq t\leq T_{3M}$, $T_{1W}\leq t\leq T_{1M}$ and $T_{1D}\leq t\leq T_{1W}$ to obtain $\sigma_i(t_0,t)$ for all $t\leq T_{6M}$.
With these five quotations for all $t$ at hand, we can recover the five corresponding strikes $K_i(t_0,t)$ and run the single maturity calibration procedure for maturity $t$ from the previous subsection.
The resulting parametrization is denoted by
\begin{equation*}
	\theta^*(t_0,t)\quad\text{for}\quad 0\leq t\leq T_{6M}.
\end{equation*}

Figure \ref{fig:Vol_surface} shows a corresponding sample volatility surface for the date 2014-01-01.
\begin{figure}[H]
    \centering
    \includegraphics[width=0.95\textwidth]{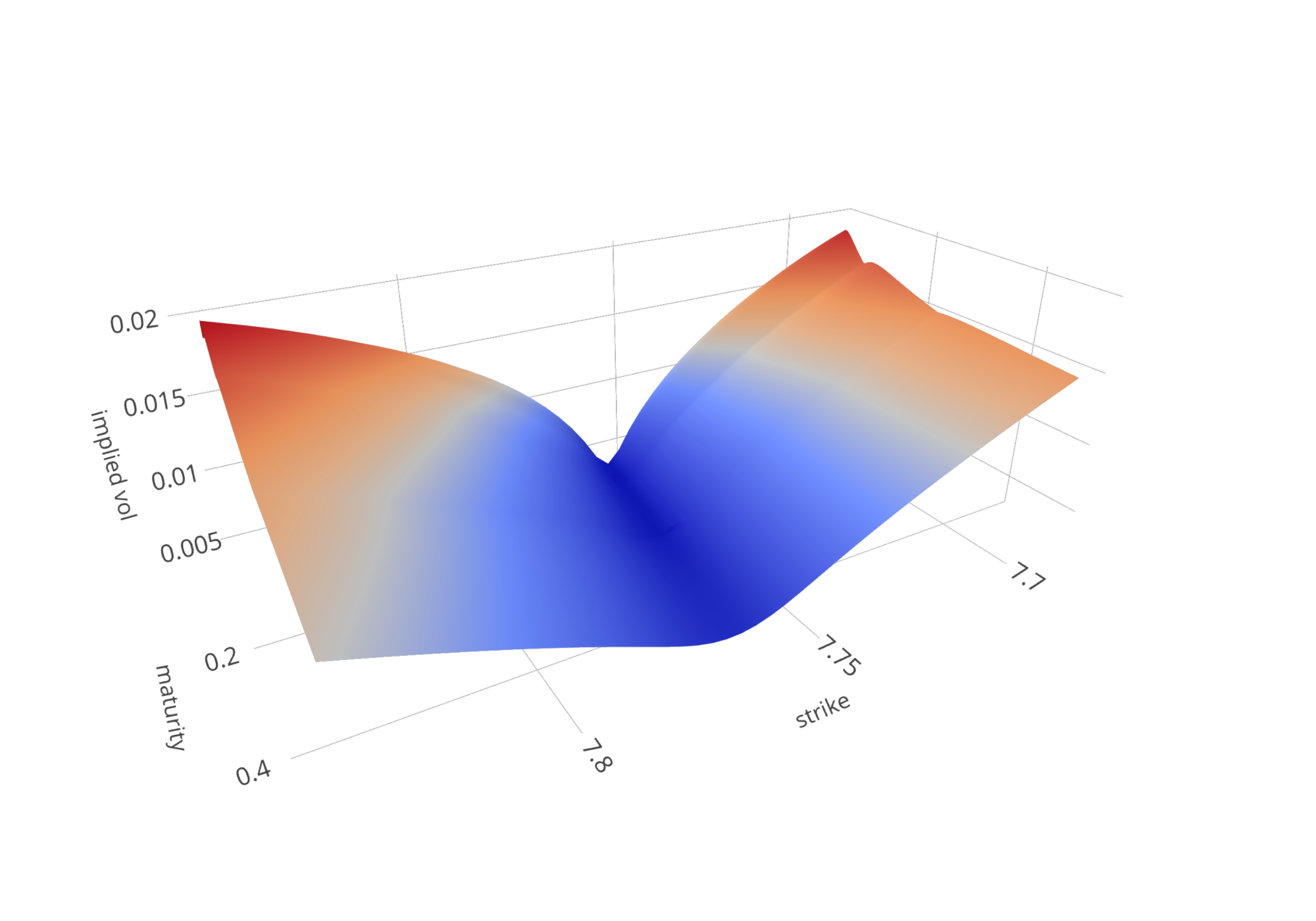}
    \caption{Volatility surface of HKDUSD option at 2014-01-01}
    \label{fig:Vol_surface}
\end{figure} 

\subsection{Computational time}
RS method which relies on a implicit computation of the implied volatility from a pricing formula that involves an integration is naturally slower than the more explicit SABR models or approximated RS.
Due to the integration part in the pricing formula, we therefore compare here after the computational time for the calibration from martingale approach versus Fourier.
The following table \ref{tab:recalibration_compute_time} shows the computational time for a parameter-calibration procedure for one-day volatility surface on 2014-01-01, therefore 131 single maturity calibrations.
All the computations were performed using Python on an Intel(R) Core(TM) i7-6500U CPU @ 2.50GHz PC with 8GB of RAM without particular optimization scheme.
The two methods use the same implied volatility function which is a plain root-finding algorithm.
\begin{table}[!ht]
    \begin{center}
        \begin{tabular}{@{}lcccr@{}}
            \toprule
         && \multicolumn{1}{c}{Martingale approach} && \multicolumn{1}{c}{Fourier approach} \\
         \midrule
            Volatility surface calibration  && 52 min 49.08 sec 	 && 9 min 59.48 sec 	 \\
            \bottomrule
        \end{tabular}
        \caption{Computational time for the calibration of $\theta^\ast(t_0,t)$, $0\leq t\leq T_{6M}$, corresponding to 131 single maturity calibrations.}
        \label{tab:recalibration_compute_time}
    \end{center}
\end{table}
The run times show that the Fourier approach is significantly faster than the martingale one by a factor of about 5.2.
Note that these 131 single maturity calibrations can be performed in parallel, amounting here for about 4 seconds.
Hence with a C++ implementation and optimization, volatility surface in the RS model with Fourier approach could be used in real time.

\section{Hedges on real data}\label{sec:real_hedge}
For the reasons presented in the previous section, we use from now on the RS model with Fourier approach to calibrate the RS parameterization.

\subsection{Hedge portfolio}
Every day, for the time period from 2014-01-01 to 2018-07-12, a 6-month ATM call is written and hedged until maturity -- the last written call expiring on 2019-01-10.
Throughout, we normalize 6 months maturity to $130$ days, which amounts therefore to 1182 delta hedges to perform.
We denote these dates by $t_n$, for $n=1,2,\cdots,1182$.
Under this discrete-time setting, timestamp $t$ is taking values in $\mathbb{T} = \{0, dt, 2dt,\cdots,130dt\}$ where $dt=1/260$.
For $n$th hedge, at date $t_n$, the trader writes a 6-month ATM call at strike price $K_{ATM}(t_n,T_{6M})$.
We compare the performance between the following three different strategies -- mean-variance being already ruled out from numerical simulation results -- where at date $t_{n+m}$, $m=0,1,\cdots,130$, the delta hedges are given by
\begin{itemize}
    \item \textbf{BS delta:} Applying the BS delta given in \eqref{equ:gkdelta}.
        This hedge takes 
        \begin{equation*}
            \sigma_{RS}\left(K_{ATM}\left(t_n,T_{6M}\right), \theta^*\left(t_{n+m},T_{6M} - mdt\right)\right)  
        \end{equation*}
        to be the BS volatility parameter.
    \item \textbf{RS delta hedging:} Applying the RS Delta given in \eqref{equ:rsdelta} with parameter
        \begin{equation*}
            \theta^\ast\left(t_{n+m}, T_{6M}-mdt\right)
        \end{equation*}
    \item \textbf{Approximated RS delta:} Applying the approximated RS delta given in \eqref{equ:approx_RSdelta} with parameter
        \begin{equation*}
            \theta^\ast\left(t_{n+m}, T_{6M}-mdt\right).
        \end{equation*}
\end{itemize}
where $\theta^\ast(t_{n+m}, s)$ for $0\leq s \leq T_{6M}$ is the calibrated RS parametrisation from the previous section using the RS model with Fourier approach.

\subsection{Result}

As for the hedging error, we still use \eqref{equ:hedge_err} and \eqref{equ:hedge_mte}.
The results of hedging errors are reported in Table(\ref{tab:heging_result}) and Figure(\ref{fig:heging_result}).
There is no much difference between these three strategies since the mean of the hedging errors are around $0.25\%$.

\begin{table}[H]
    \begin{center}
        \resizebox{\columnwidth}{!}{
            \begin{tabular}{@{}llrrrrrrr@{}}
                \toprule
                \multicolumn{1}{c}{} &\multicolumn{1}{c}{} & \multicolumn{1}{c}{Mean} & \multicolumn{1}{c}{Std} & \multicolumn{1}{c}{Min}& \multicolumn{1}{c}{q-25\%}& \multicolumn{1}{c}{Median}& \multicolumn{1}{c}{q-75\%}& \multicolumn{1}{c}{Max}\\
                \midrule
                &BS&0.252\%    &0.168    &0.001\%    &0.077\%    &0.239\%    &0.385\%    &0.746\%\\
$\text{Error}_{hedge}$&RS&0.291\%    &0.214   &0.000\%    &0.071\%    &0.317\%    &0.459\%    &0.863\%\\
                &Approx RS&0.291\%    &0.213   &0.000\%    &0.071\%    &0.318\%    &0.457\%    &0.863\%\\

                \midrule

                &BS&0.130\%    &0.098    &0.011\%    &0.046\%    &0.121\%    &0.186\%    &0.524\%\\
$\text{MTError}_{hedge}$&RS&0.152\%    &0.115    &0.010\%    &0.048\%    &0.134\%    &0.224\%    &0.588\%\\
                &Approx RS&0.152\%    &0.116    &0.010\%    &0.048\%    &0.135\%    &0.224\%    &0.590\%\\
                
                \bottomrule
            \end{tabular}
        }
        \caption{Hedging errors: HKDUSD 6M ATM Call}
        \label{tab:heging_result}
    \end{center}
\end{table}

\begin{figure}[H]
    \centering
    \includegraphics[width=\textwidth]{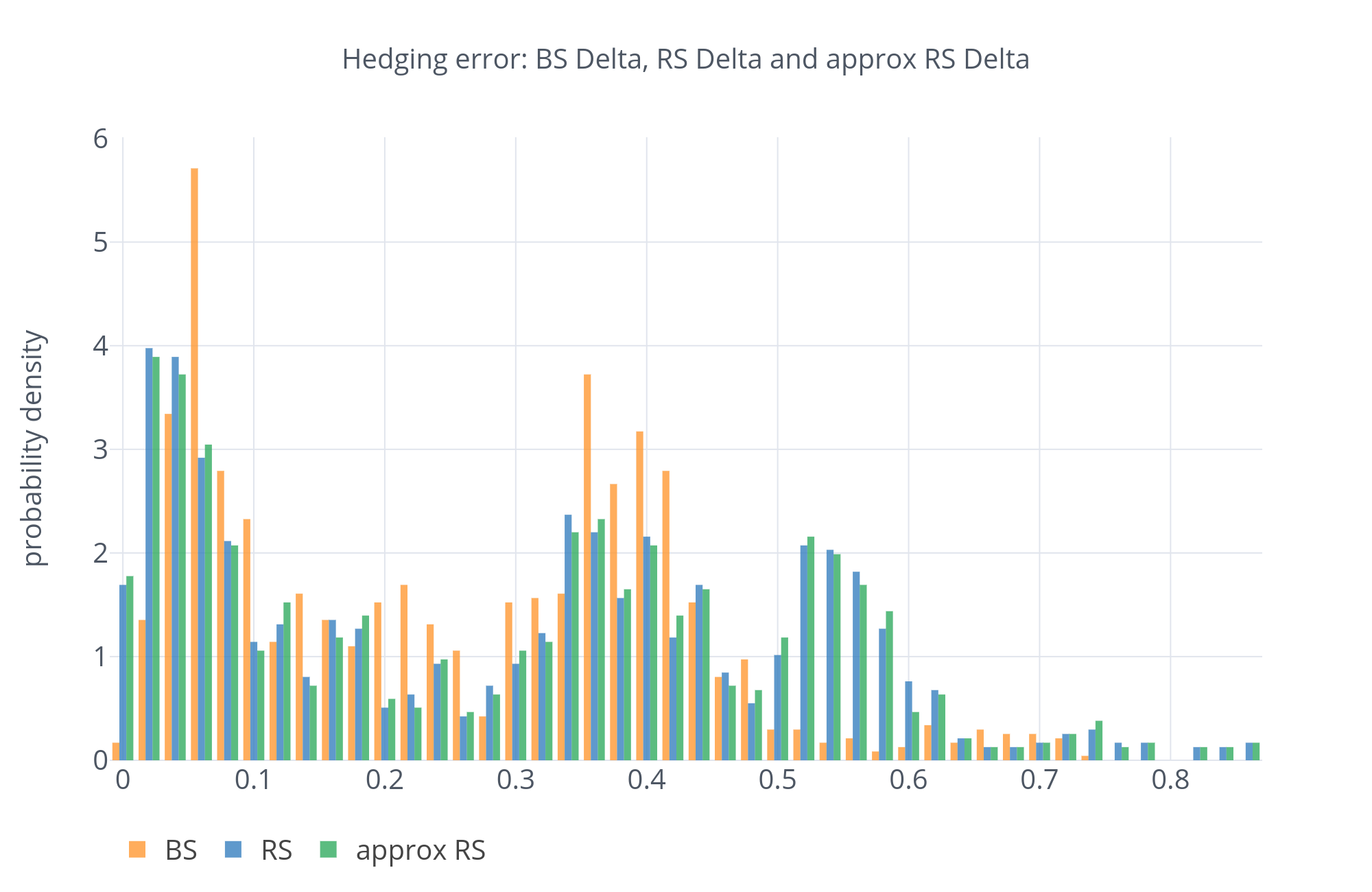}
    \caption{Hedging error: HKDUSD 6M ATM Call}
    \label{fig:heging_result}
\end{figure}

\begin{figure}[H]
    \centering
    \includegraphics[width=\textwidth]{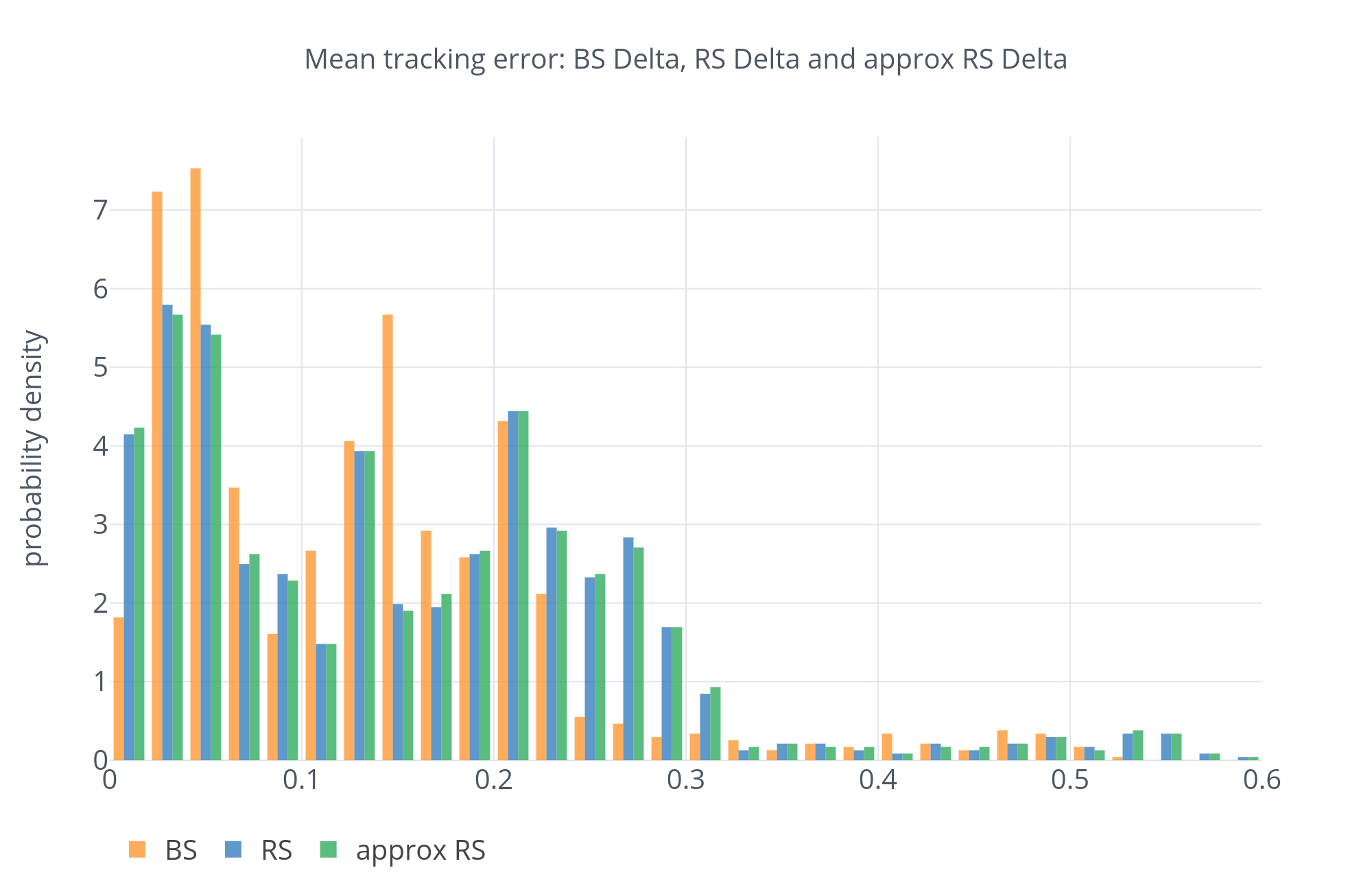}
    \caption{Mean tracking error: HKDUSD 6M ATM Call}
    \label{fig:heging_result_m}
\end{figure}
Some selected hedging performance of hedging paths for BS Delta, exact RS Delta and approximated RS Delta hedges respectively are illustrated in Figure \ref{fig:hedges_selected}.
Here we choose the HKDUSD 6M ATM call options starting on the dates 2014-01-02, 2014-10-08, 2015-12-02 and 2017-11-01.
\begin{figure}[H]
    \begin{center}
        \begin{subfigure}[b]{0.49\textwidth}
            \includegraphics[width=\textwidth]{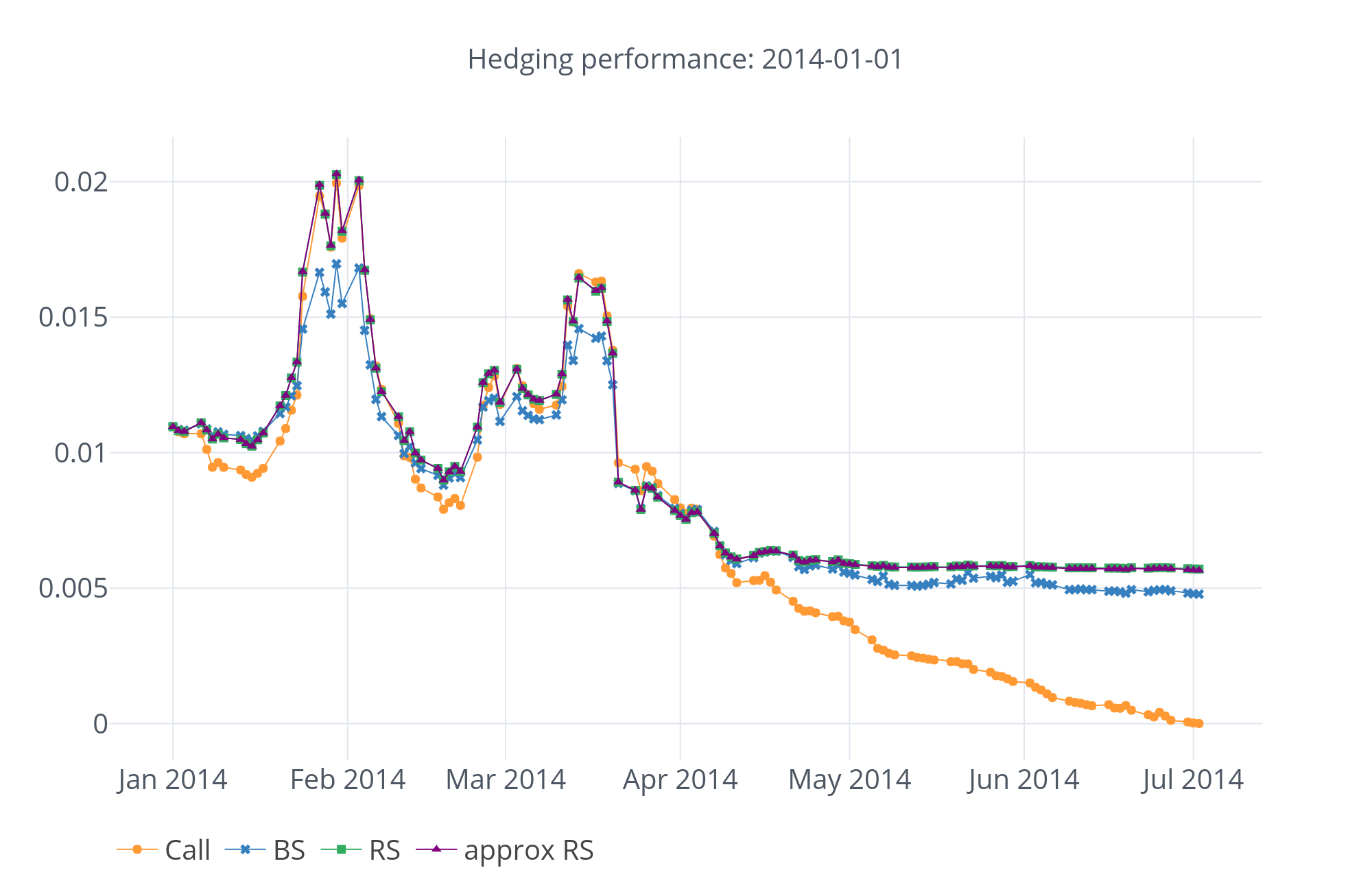}
        \end{subfigure}
        \hfill
        \begin{subfigure}[b]{0.49\textwidth}
            \includegraphics[width=\textwidth]{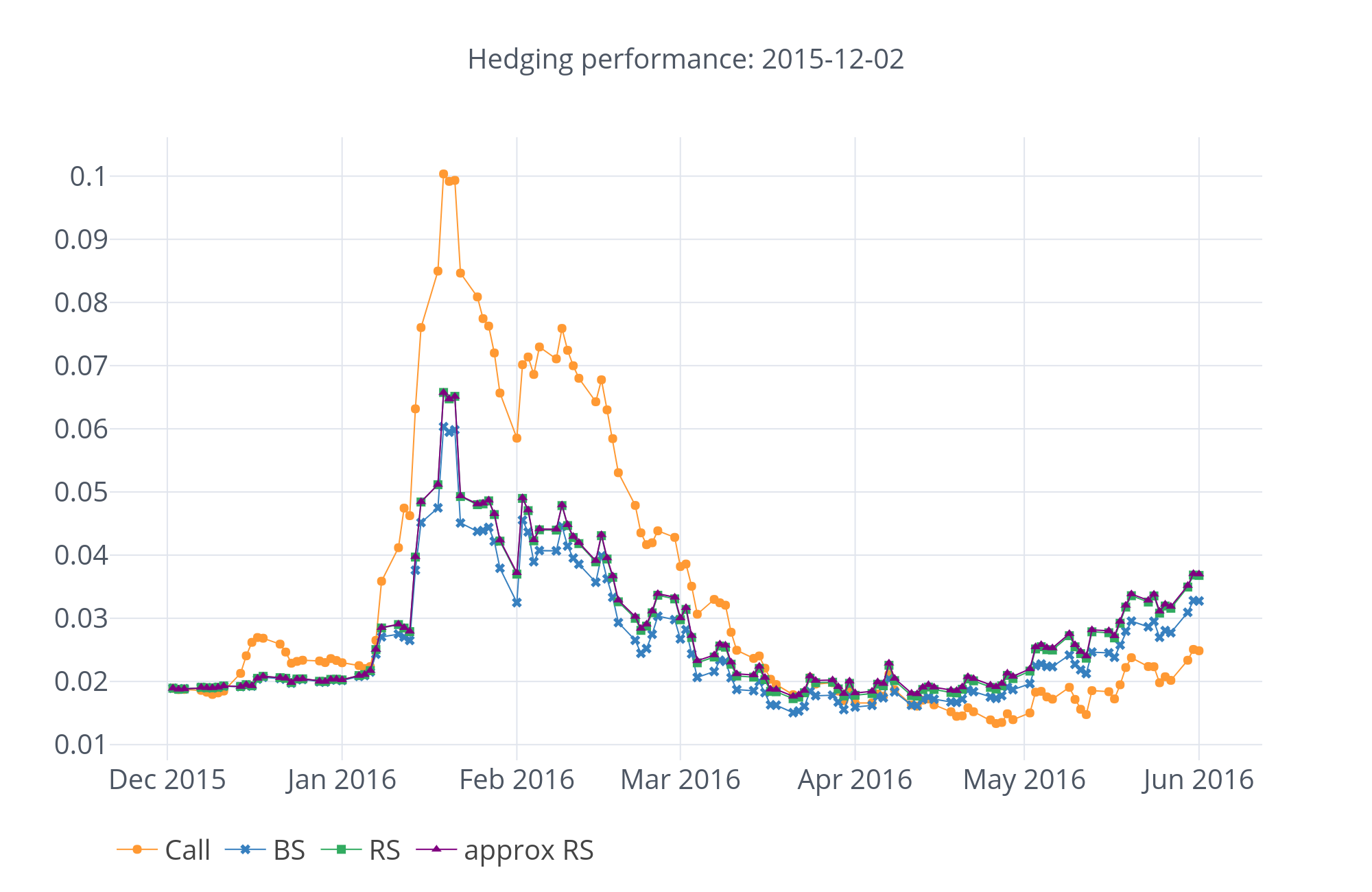}
        \end{subfigure}
    \end{center}
\end{figure}
\begin{figure}[H]\ContinuedFloat
    \begin{center}
        \begin{subfigure}[b]{0.49\textwidth}
            \includegraphics[width=\textwidth]{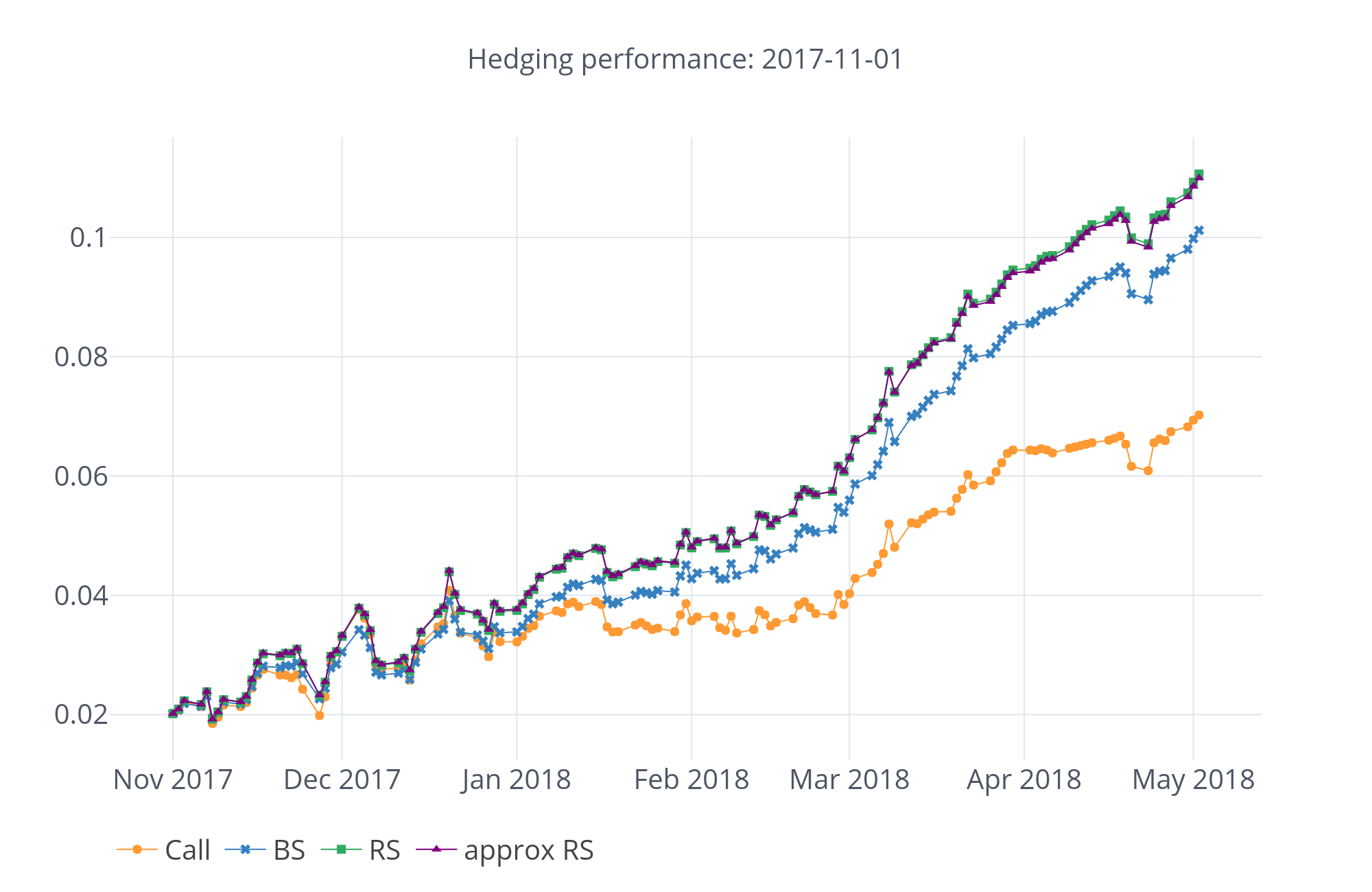}
        \end{subfigure}
        \hfill
        \begin{subfigure}[b]{0.49\textwidth}
            \includegraphics[width=\textwidth]{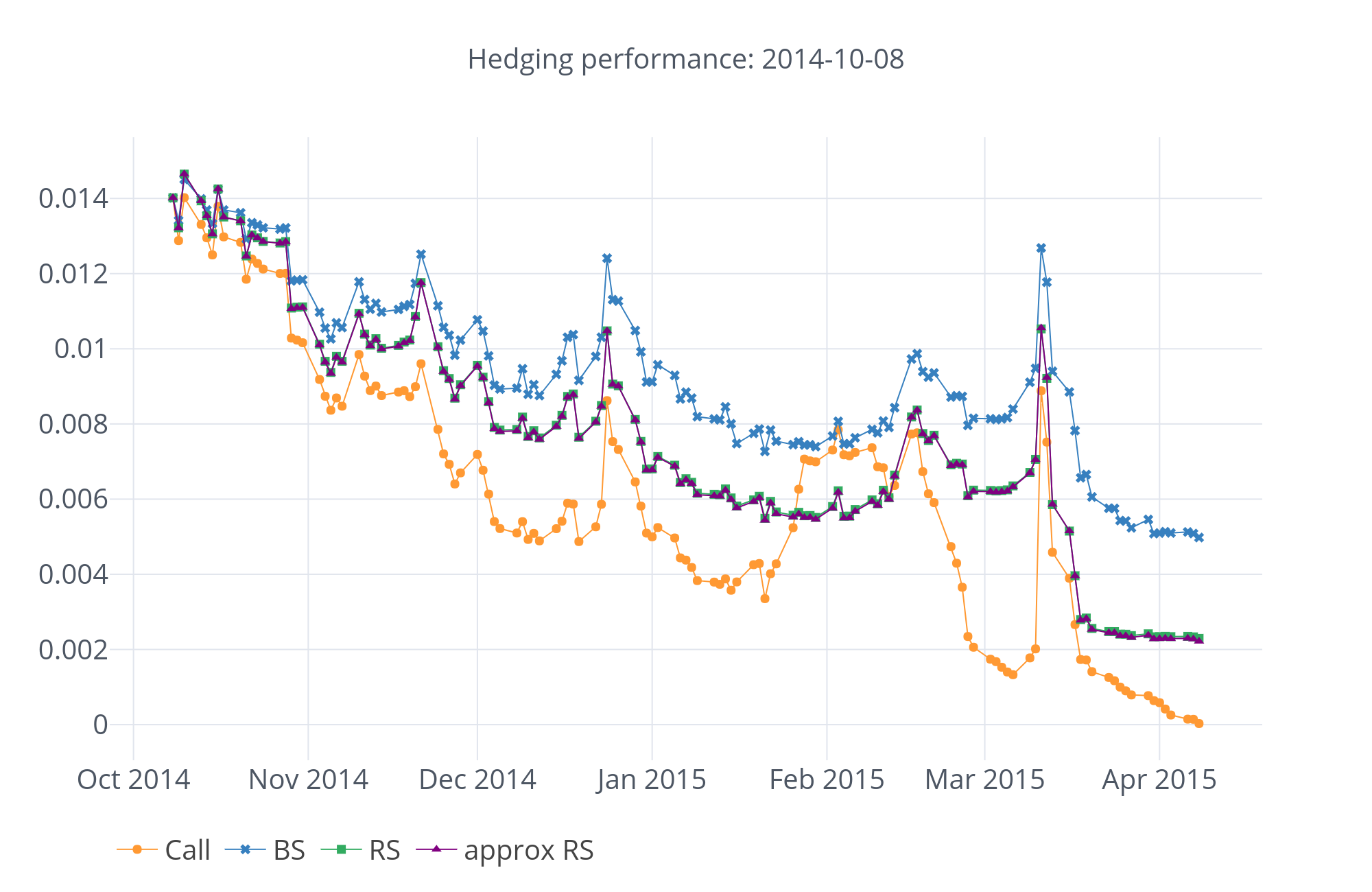}
        \end{subfigure}
    \end{center}
    \caption{Hedging performance for selected dates}
    \label{fig:hedges_selected}
\end{figure}
The Figure \ref{fig:hedges_selected} shows that the approximated RS hedge have roughly the same performance as the exact RS hedge.

\section{Conclusion}\label{sec:conclution}
In this paper, we investigate the pricing/hedging accuracy and performance based on a specific regime switching model for pegged markets.

As for the pricing, we compare two prices, an exact pricing formula and a first order approximation.
For the first order approximation, we provide the theoretical bonds for the errors.
Based on real data for HKDUSD, we compare the performance of each pricer for the calibration of the volatility surface using as a benchmark the classical SABR model.
It turns out that the approximated pricing formula is inaccurate.
As for the exact pricing formula, it is very accurate and outperforms SABR model almost all the time.

As for the hedging, we compare exact RS delta, approximated RS delta as well as mean-variance hedging strategies.
Based on simulated data set, the hedging strategies are tested by applying BS delta, RS delta, approximated RS delta, mean-variance and approximated mean-variance. 
Mean-variance strategies produce large errors at "no-jump" scenario, while the improvement is not significant at "jump" scenario.
Therefore mean-variance hedging is not appropriate in this context.
Based on real data set for HKDUSD, we apply BS delta, RS delta, approximated RS delta hedging strategies.
It turns out that the hedging error of the approximated RS delta hedge does not differ much the two others, while the computational time of the approximated RS delta is significantly faster than the exact one.

Concerning the performance, though the exact pricing formula calibrates almost perfectly, it is computationally intensive.
To overcome that problem, we provide an explicit formula for the moment generating function, and compare the performance with respect to the direct pricing method.
As expected the Fourier approach is faster by a factor of over 5.

\newpage
\begin{footnotesize}
\begin{appendix}
    \section{SABR Parametrization}\label{app:SABR}
    As a benchmark, we use the SABR model by \citet{hagan2002}, where the implied volatility is explicitly given by
    \begin{equation*}
        \sigma_{SABR}(K;\tilde\theta) = a\frac{z}{\chi(z)}\left(1 + \left(\frac{\rho b a}{4} + \frac{2-3\rho^2}{24} b^2\right)T\right),
    \end{equation*}
    where 
    \begin{equation*}
        \chi(z) = \ln\left(\frac{\sqrt{1-2\rho z+z^2} + z - \rho}{1-\rho}\right),\quad z=\frac{ b}{a}\ln(F(0,T)/K)\quad\text{and}\quad \tilde\theta = (a, b,\rho),
    \end{equation*}
    with initial volatility $a$,  volatility of the volatility $b$, correlation between the two different Brownian motions $\rho$ and forward price $F(0,T)$.

    \section{Fourier approach}\label{app:Fourier}
    \begin{proof}
        For $\xi(\alpha) := -\sigma^2(\alpha)/2 - \lambda\kappa(\alpha)$, it follows that
        \begin{equation*}
            X(T) = \begin{cases}
                \xi(0)\tau + \xi(1)(T-\tau) + \ubar{\sigma}W(\tau)+\bar{\sigma}W(T-\tau) + Y, &\text{ if }\tau\leq T,\\
                \xi(0)T +  \ubar{\sigma}W(T), &\text{ if }\tau> T.\\
            \end{cases}
        \end{equation*}
        where $Y \sim \mathcal{N}(u, \delta^2)$.
        Since $\tau=\inf\{t:M(t)>0\}$, the distribution of $\tau$ is given by $\lambda e^{-\lambda t}dt$.
        Defining $\mathcal{F}^\alpha = \sigma\{\alpha(s) \colon 0\leq s \leq T\}$, the characteristic function is given by
        \begin{align*}
            E\left[e^{izX(T)}\right] & = E\left[E\left[e^{izX(T)}\big|\mathcal{F}^\alpha\right]\right]\\
                                     & = E\left[E\left[e^{izX(T)}\big|\tau\leq T\right]1_{\{\tau\leq T\}}\right]+E\left[E\left[e^{izX(T)}\big|\tau>T\right]1_{\{\tau>T\}}\right].
        \end{align*}
        As for the second term on the right hand side, on the event $\{\tau>T\}$, from $X(T)\sim \mathcal{N}(\xi(0)T,\ubar{\sigma}^2T)$ follows that 
        \begin{equation*}
            E\left[E\left[e^{izX(T)}\big|\tau\right]1_{\{\tau>T\}}\right]=e^{iz \xi(0)T -z^2 \frac{\ubar{\sigma}^2T}{2}}e^{-\lambda T}=e^{-\lambda T}\phi_{0,T}(z)
        \end{equation*}
        As for the first term, on the event $\{\tau \leq T\}$, from $X(T) \sim \mathcal{N}(\xi(0)\tau+\xi(1)(T-\tau)+u, \ubar{\sigma}^2\tau+\bar{\sigma}^2(T-\tau)+\delta^2)$, we get
        \begin{multline*}
            E\left[E\left[e^{izX(T)}\big|\tau\right]1_{\{\tau\leq T\}}\right]
            = \int_0^Te^{iz\left(\xi(0)t+\xi(1)(T-t)+u\right)-\frac{1}{2}z^2\left(\ubar{\sigma}^2t+\bar{\sigma}^2(T-t)+\delta^2\right)}\lambda e^{-\lambda t}dt\\
            = \lambda e^{izu-z^2\frac{\delta^2}{2}} e^{-iz \frac{\bar{\sigma}^2T}{2}-z^2\frac{\bar{\sigma}^2 T}{2}}\int_0^Te^{\left(iz(\xi(0)-\xi(1))-\frac{1}{2}z^2(\ubar{\sigma}^2-\bar{\sigma}^2)-\lambda\right)t}dt\\
            = \frac{\lambda e^{izu-z^2\frac{\delta^2}{2}} e^{-iz \frac{\bar{\sigma}^2T}{2}-z^2\frac{\bar{\sigma}^2 T}{2}}\left(e^{\left(iz(\xi(0)-\xi(1))-\frac{1}{2}z^2(\ubar{\sigma}^2-\bar{\sigma}^2)-\lambda\right)T}-1\right)}{\left(iz+z^2\right)\frac{\bar{\sigma}^2 - \ubar{\sigma}^2}{2} - \lambda(izk+1)}\\
            = \lambda e^{izu-z^2\frac{\delta^2}{2}}\frac{e^{-\lambda T}\phi_{0,T}(z) - \phi_{1,T}(z)}{\left(iz+z^2\right)\frac{\bar{\sigma}^2 - \ubar{\sigma}^2}{2} - \lambda(izk+1)}
        \end{multline*}
        which ends the proof.
    \end{proof}

    \section{First Order Approximation}\label{app:first_order_approx}
    \begin{proof}[Proof of Proposition \ref{prop:approx}]
        Note that
        \begin{equation*}
            \left| V(\theta)-V_{approx}(\theta) \right| = \left| \int_{0}^{T}\left( f(0)-f(t) \right) \lambda e^{-\lambda t} dt  \right| = \left| \int_{0}^{T}\int_{0}^{t} f^{\prime}(s) \lambda e^{-\lambda t}ds dt   \right|
        \end{equation*}
        Where $f(t) = GK(S_0 e^{-\lambda \kappa t}(1+\kappa), \sqrt{\ubar{\sigma}^2 t+\bar{\sigma}^2(T-t)})$.
        However
        \begin{multline*}
        \left|f^\prime(s) \right|= \left|S_0(1+\kappa)\frac{de^{-\lambda \kappa s}}{ds}\Delta_{GK}\left(S_0 e^{-\lambda \kappa s}(1+\kappa), \sqrt{\ubar{\sigma}^2 s+\bar{\sigma}^2(T-s)}\right)\right .\\
        \left. +\frac{d\sqrt{\ubar{\sigma}^2 s+\bar{\sigma}^2(T-s)}}{ds}{\frac{1}{\sqrt{T}}}  \mathcal{V}_{GK} \left(S_0 e^{-\lambda \kappa s}(1+\kappa), \sqrt{\ubar{\sigma}^2 s+\bar{\sigma}^2(T-s)}\right)\right|\\
            \leq S_0(1+\kappa)|\kappa|\lambda e^{-\lambda \kappa s} \Delta_{GK}\left(S_0 e^{-\lambda \kappa s}(1+\kappa), \sqrt{\ubar{\sigma}^2 s+\bar{\sigma}^2(T-s)}\right)\\
            + \left| \frac{d\sqrt{\ubar{\sigma}^2 s+\bar{\sigma}^2(T-s)}}{ds}\right|\frac{1}{\sqrt{T}}  \mathcal{V}_{GK} \left(S_0 e^{-\lambda \kappa s}(1+\kappa), \sqrt{\ubar{\sigma}^2 s+\bar{\sigma}^2(T-s)}\right)
        \end{multline*}
        where
        \begin{multline}\label{equ:f_prime_kappa}
            S_0(1+\kappa)|\kappa|\lambda e^{-\lambda \kappa s} \Delta_{GK}\left(S_0 e^{-\lambda \kappa s}(1+\kappa), \sqrt{\ubar{\sigma}^2 s+\bar{\sigma}^2(T-s)}\right)\\
            \leq 
            \begin{cases}
                \lambda \kappa S_0(1+\kappa)e^{-\lambda \kappa s},\quad\text{if }\kappa \geq 0\\
                \\
                \lambda |\kappa| S_0(1+\kappa)e^{-\lambda \kappa s},\quad\text{if }-1<\kappa < 0
            \end{cases}
        \end{multline}
        and
        \begin{multline}\label{equ:comb_sigma}
            \left|\frac{d\sqrt{\ubar{\sigma}^2 s+\bar{\sigma}^2(T-s)}}{ds}\right|\frac{1}{\sqrt{T}} \mathcal{V}_{GK} \left(S_0 e^{-\lambda \kappa s}(1+\kappa), \sqrt{\ubar{\sigma}^2 s+\bar{\sigma}^2(T-s)}\right)\\
            = \left|\frac{d\sqrt{\ubar{\sigma}^2 s+\bar{\sigma}^2(T-s)}}{ds}\right| S_0e^{-r_fT}N'(d_+)
            \leq \left|\frac{d\sqrt{\ubar{\sigma}^2 s+\bar{\sigma}^2(T-s)}}{ds}\right|\frac{S_0}{\sqrt{2\pi}}.
        \end{multline}
        From \eqref{equ:f_prime_kappa} and \eqref{equ:comb_sigma} and using
        \begin{equation*}
            \int_0^t \left|\frac{d\sqrt{\ubar{\sigma}^2 s+\bar{\sigma}^2(T-s)}}{ds} \frac{1}{\sqrt{2\pi}}\right|ds
            = \frac{1}{\sqrt{2\pi}}(\bar{\sigma}\sqrt{T} - \sqrt{\ubar\sigma^2t+\bar\sigma^2(T-t)})
            \leq \sqrt{\frac{T}{2\pi}}(\bar{\sigma} - \ubar{\sigma}).
        \end{equation*}
        yields
        \begin{equation*}
            \frac{1}{S_0}\int_0^t|f'(s)|ds \leq 
            \begin{cases}
                (1+\kappa)(1-e^{-\lambda \kappa t}) + \sqrt{\frac{T}{2\pi}}(\bar{\sigma} - \ubar{\sigma})
        & \text{if }\kappa \geq 0\\
        \\
        (1+\kappa)(e^{-\lambda \kappa t}-1) + \sqrt{\frac{T}{2\pi}}(\bar{\sigma} - \ubar{\sigma})
        & \text{if }-1<\kappa < 0
            \end{cases}
        \end{equation*}
        We obtain for $\kappa\geq 0$:
        \begin{align*}
            \frac{\left| V(\theta)-V_{approx}(\theta) \right|}{S_0} & \leq \int_{0}^{T} \left[(1+\kappa)(1-e^{-\lambda \kappa t}) + \sqrt{\frac{T}{2\pi}}(\bar{\sigma} - \ubar{\sigma})\right]\lambda e^{-\lambda t} dt \\
            &\leq (1-p)\left(1+\kappa + \sqrt{\frac{T}{2\pi}}(\bar{\sigma} - \ubar{\sigma})\right) - \left(1-e^{-\lambda(1+\kappa)T}\right)
        \end{align*}
        and for $-1<\kappa<0$:
        \begin{align*}
            \frac{\left| V(\theta)-V_{approx}(\theta) \right|}{S_0} & \leq \int_{0}^{T} \left[(1+\kappa)(e^{-\lambda \kappa t}-1) + \sqrt{\frac{T}{2\pi}}(\bar{\sigma} - \ubar{\sigma})\right]\lambda e^{-\lambda t} dt \\
             & \leq \left(1-e^{-\lambda(1+\kappa)T}\right) + (1-p)\left( \sqrt{\frac{T}{2\pi}}\left(\bar{\sigma} - \ubar{\sigma}\right)- (1+\kappa)\right).
        \end{align*}
        Together, it holds
        \begin{equation*}
            \frac{\left| V(\theta)-V_{approx}(\theta) \right|}{S_0}  \leq (1-p)\sqrt{\frac{T}{2\pi}}(\bar{\sigma} - \ubar{\sigma}) + |\kappa|(1-p) - p\left|e^{-\lambda\kappa T} - 1\right|.
        \end{equation*}
    \end{proof}

    \section{Mean-variance hedging strategy}\label{app:MV}
    In the present model, the Markov chain process $\alpha$ takes values in the state space $\{0, 1\}$ with infinitesimal generator:
    \begin{equation*}
        Q = 
        \begin{bmatrix}
            -\lambda & \lambda  \\
            0        & 0	     
        \end{bmatrix}. 
    \end{equation*}
    In particular,
    \begin{equation*}
        \alpha(t) = \alpha(0) + \int_0^t\int_\mathbb{R}h(\alpha(s-),y)M(dy,ds)
    \end{equation*}
    where 
    \begin{equation*}
        h(i,y):=
        \begin{cases}
            j-i,\quad &\text{if } y\in\Lambda_{ij}\\
            0,\quad   &\text{otherwise}
        \end{cases}
    \end{equation*}
    where $i,j\in\{0, 1\}$ and $\Lambda_{ij}$ is a consecutive left closed right open intervals of the real line with length $q_{ij}$ ($q_{ij}$ is the element in $Q$), see \citep{basak2011} and \citep{shaoyong2017}.
	We further define $v(dy)$ as the L\'evy measure and $m(dy, dt)$ as the compensator measure of the Poisson random measure $M(dy,dt)$.
    For any Borel set $A\subseteq \mathbb{R}\setminus 0$,
    \begin{equation*}
        \tilde M(A,(0,t]) = M(A,(0,t]) - m(A,(0,t])
    \end{equation*}
    is a $P$-martingale.
    In this model, $v(dy)=\lambda N_{u, \delta^2}^{\prime}(y)dy$ and $m(dy,dt) = v(dy)dt = \lambda N^{\prime}_{u, \delta^2}(y)dydt$ where $N_{u, \delta^2}$ is the CDF of a normal distribution $\mathcal{N}(u, \delta^2)$.

    \subsection{mean-variance hedging}
    In this setting, $\tilde S(t) := S(t)/e^{(r_d - r_f)t}$ is a square integrable martingale under $P$, see \citet{anindya2019} for instance.
    Mean variance hedging aims at finding the initial capital $C_0$ and the predictable portfolio $\pi_t$ such that 
    \begin{equation*}
        \inf_{C_0,\pi}E\left[\left|\tilde C_T -\tilde H\right|^2\right], \text{ where } \tilde C_T = C_0 +  \int_0^T\pi_td\tilde S(t),
    \end{equation*}
    where $\tilde H:= H/e^{(r_d - r_f)T}$ is the discounted payoff\footnote{The payoff $H$ is a function of $S(T)$, i.e., $H=H(S(T))$.}.
    Assume that $H\in L^2(\Omega,\mathcal{F},P)$ is a square integrable random variable and the terminal values of the portfolios have finite variance:
    \begin{equation}\label{equ:pi_integrable}
        E\left[\left|\int_0^T \pi_td\tilde S(t)\right|^2\right] <\infty.
    \end{equation}
    According to \citep[chapter 10.4]{tankov2004}, \eqref{equ:pi_integrable} is equivalent to
    \begin{equation*}
        E\left[\int_0^T |\pi_t\tilde S_t|^2dt + \int_0^T\int_{\mathbb{R}}\left(e^{\gamma(y,\alpha(t-))}-1\right)^2|\pi_t\tilde S_t|^2v(dy)dt\right]<\infty.
    \end{equation*}
    Further, we assume
    \begin{equation}\label{equ:Sjump_integrable}
        E\left[\int_0^T\int_{\mathbb{R}}\left(S(t-)e^{\gamma(y,\alpha(t-))}\vee S(t-)\right)^2v(dy)dt\right]<\infty.
    \end{equation}
    The optimal value of the initial capital is given by $C_0 = E[\tilde H]$ which corresponds to the price in the RS model.

    \begin{proof}
        The proof follows \citet[Chapter 1.4]{tankov2004}.
        Consider a self-financing trading strategy $\pi$, the terminal value of which is given by 
        \begin{equation}\label{equ:terminal_value}
            \int_0^T \pi_t d\tilde S(t) =\int_0^T \pi_t\tilde S(t-)\sigma(\alpha(t-))dW_t + \int_0^T\int_\mathbb{R} \left(e^{\gamma(y,\alpha(t-))}-1\right)\pi_t\tilde S(t-) \tilde M(dt,dy).
        \end{equation}
        For European call option, the payoff is $H(S(T)) = (S(T)-K)^+$.
        Recall that the functions $\tilde {\mathcal{C}}(t,s,x) = e^{-(r_d-r_f)t}{\mathcal{C}}(t,s,x)$.
        From the construction, $\tilde {\mathcal{C}}(t,S(t),\alpha(t))$ is a martingale.
        Applying It\^o formula, it follows
        \begin{multline}\label{equ:C_tilde_ito}
            \tilde{\mathcal{C}}(t,S(t),\alpha(t)) -\tilde{\mathcal{C}}(0,S(0),\alpha(0))\\
            = \int_0^t\frac{\partial \tilde{\mathcal{C}}}{\partial u}(u,S(u-),\alpha(u-))du+\int_0^t\frac{\partial \tilde{\mathcal{C}}}{\partial s}(u,S(u-),\alpha(u-))S(u-)(r_d-r_f-\lambda\kappa(\alpha(u-))) du\\
            + \int_0^t\frac{\partial \tilde{\mathcal{C}}}{\partial s}(u,S(u-),\alpha(u-))S(u-)\sigma(\alpha(u-)) dW(u)\\ 
            +\int_0^t\frac{1}{2}\frac{\partial^2 \tilde{\mathcal{C}}}{\partial s^2}(u,S(u-),\alpha(u-))\sigma^2(\alpha(u-))S^2(u-)du\\
            + \int_0^t\int_\mathbb{R}\left[\tilde{\mathcal{C}}\left(u,S(u-)e^{\gamma(y,\alpha(t-))},\alpha(u-)+h(\alpha(t-))\right) - \tilde{\mathcal{C}}(u,S(u-),\alpha(u-))\right]M(dy,du)\\
            = \int_0^t\frac{\partial \tilde{\mathcal{C}}}{\partial t}(u,S(u-),\alpha(u-))du+\int_0^t\frac{\partial \tilde{\mathcal{C}}}{\partial s}(u,S(u-),\alpha(u-))S(u-)(r_d-r_f-\lambda\kappa(\alpha(u-))) du\\
            + \int_0^t\frac{\partial \tilde{\mathcal{C}}}{\partial s}(u,S(u-),\alpha(u-))S(u-)\sigma(\alpha(u-)) dW(u)\\ 
            +\int_0^t\frac{1}{2}\frac{\partial^2 \tilde{\mathcal{C}}}{\partial s^2}(u,S(u-),\alpha(u-))\sigma^2(\alpha(u-))S^2(u-)du\\
            + \int_0^t\int_\mathbb{R}\left[\tilde{\mathcal{C}}\left(u,S(u-)e^{\gamma(y,\alpha(t-))},\alpha(u-)+h(\alpha(t-),y)\right) - \tilde{\mathcal{C}}(u,S(u-),\alpha(u-))\right]\tilde M(dy,du)\\
            +\int_0^t\int_\mathbb{R}\left[\tilde{\mathcal{C}}\left(u,S(u-)e^{\gamma(y,\alpha(t-))},\alpha(u-)+h(\alpha(t-),y)\right) - \tilde{\mathcal{C}}(u,S(u-),\alpha(u-))\right]v(dy)du
        \end{multline}
        Since the fact that $\tilde{\mathcal{C}}(t,S(t),\alpha(t))$ is a martingale which of the form
        \begin{multline*}
            \tilde{\mathcal{C}}(t,S(t),\alpha(t)) -\tilde{\mathcal{C}}(0,S(0),\alpha(0))\\
            = \int_0^t\frac{\partial \tilde{\mathcal{C}}}{\partial S}(u,S(u-),\alpha(u-))S(u-)\sigma(\alpha(u-)) dW(u)\\
            + \int_0^t\int_\mathbb{R}[\tilde{\mathcal{C}}\left(u,S(u-)e^{\gamma(y,\alpha(t-))},\alpha(u-)+h(\alpha(u-),y)\right) - \tilde{\mathcal{C}}(u,S(u-),\alpha(u-))]\tilde M(dy,du)\\
            = \int_0^t\frac{\partial \mathcal{C}}{\partial S}(u,S(u-),\alpha(u-))\tilde S(u-)\sigma(\alpha(u-)) dW(u)\\
            + \int_0^t\int_\mathbb{R}[\tilde{\mathcal{C}}\left(u,S(u-)e^{\gamma(y,\alpha(t-))},\alpha(u-)+h(\alpha(u-),y)\right) - \tilde{\mathcal{C}}(u,S(u-),\alpha(u-))]\tilde M(dy,du).
        \end{multline*}
        Since the payoff function is Lipschitz, we have
        \begin{multline*}
            \tilde{\mathcal{C}}(t,x,1) - \tilde{\mathcal{C}}(t,y,0) = e^{-(r_d-r_f)t}\left(\mathcal{C}(t,x,1) - \mathcal{C}(t,y,0)\right)\\
            = e^{-(r_d-r_f)t}E\left[H\left(xe^{-\int_t^T\frac{\bar\sigma^2}{2}ds + \int_t^T\bar\sigma dW(s)}\right)\right.\\
            \left.- H\left(ye^{-\int_t^T\left(\frac{\bar\sigma^2}{2}+\lambda\kappa(\alpha(s-))\right)ds + \int_t^T\bar\sigma dW(s)+ \int_t^T\int_{\mathbb{R}}\gamma(y,\alpha(s-))M(dy,ds)}\right)\right]\\
            \leq KE\left[\left| xe^{-\int_t^T\frac{\bar\sigma^2}{2}ds + \int_t^T\bar\sigma dW(s)}
        \right. \right.
        \\
        \left. \left.
            - ye^{-\int_t^T\left(\frac{\bar\sigma^2}{2}+\lambda\kappa(\alpha(s-))\right)ds + \int_t^T\sigma(\alpha(s-)) dW(s)+ \int_t^T\int_{\mathbb{R}}\gamma(y,\alpha(s-))M(dy,ds)}\right |\right]\\
            \leq K\left(xE\left[e^{-\int_t^T\frac{\bar\sigma^2}{2}ds + \int_t^T\bar\sigma dW(s)}\right]\right.\\
            \left. + yE\left[e^{-\int_t^T\left(\frac{\bar\sigma^2}{2}+\lambda\kappa(\alpha(s-))\right)ds + \int_t^T\sigma(\alpha(s-)) dW(s)+ \int_t^T\int_{\mathbb{R}}\gamma(y,\alpha(s-))M(dy,ds)}\right]\right)\\
            =K(c_1x + c_2y)
        \end{multline*}
        where $K$, $c_1$ and $c_2$ are positive constants and the last equality holds from the fact that
        \begin{equation*}
            \exp\left(-\int_t^T\frac{\bar\sigma^2}{2}ds + \int_t^T\bar\sigma dW(s)\right)
        \end{equation*}
        and 
        \begin{equation*}
            \exp\left(-\int_t^T\left(\frac{\bar\sigma^2}{2}+\lambda\kappa(\alpha(s-))\right)ds + \int_t^T\sigma(\alpha(s-)) dW(s)+ \int_t^T\int_{\mathbb{R}}\gamma(y,\alpha(s-))M(dy,ds)\right)
        \end{equation*}
        are martingales.
        Therefore the predictable random function
        \begin{equation*}
            g(t,y) = \tilde{\mathcal{C}}\left(t,S(t-)e^{\gamma(y,\alpha(t-))},\alpha(t-)+h(\alpha(t-),y)\right) - \tilde{\mathcal{C}}\left(t,S(t-),\alpha(t-)\right)
        \end{equation*}
        verifies
        \begin{multline*}
            E\left[\int_0^T\int_{\mathbb{R}}|g(t,y)|^2v(dy)dt\right]\\
            =E\left[\int_0^T\int_{\mathbb{R}}\left|\tilde{\mathcal{C}}\left(t,S(t-)e^{\gamma(y,\alpha(t-))},\alpha(t-)+h(\alpha(t-),y)\right) - \tilde{\mathcal{C}}\left(t,S(t-),\alpha(t-)\right)\right|^2v(dy)dt\right]\\
            \leq E\left[\int_0^T\int_{\mathbb{R}}\left|k_1S(t-)e^{\gamma(y,\alpha(t-))}+k_2S(t-)\right|^2v(dy)dt\right]<\infty.
        \end{multline*}
        Note that the last inequality holds from assumption \eqref{equ:Sjump_integrable}.
        Hence $\tilde{\mathcal{C}}(t,S(t),\alpha(t))$ is a square integrable martingale.
        Subtracting \eqref{equ:terminal_value} from \eqref{equ:C_tilde_ito}, the hedging error is given by
        \begin{multline*}
            \epsilon(\pi) 
            = \int_0^T\left( \left(\tilde S(t-)\frac{\partial \tilde{\mathcal{C}}}{\partial S}(t,S(t-),\alpha(t-)) - \pi_t\tilde S(t-)\right)\sigma(\alpha(t-))\right)dW(t)\\
            +\int_0^T\int_\mathbb{R}\left[\tilde{\mathcal{C}}\left(t,S(t-)e^{\gamma(y,\alpha(t-))},\alpha(t-)+h(\alpha(t-),y)\right) - \tilde{\mathcal{C}}(t,S(t-),\alpha(t-)) \right.\\
            \left.- \pi_t\left(e^{\gamma(y,\alpha(t-))}-1\right)\tilde S(t-)\right]\tilde{M}(dt,dy).
        \end{multline*}

        where each stochastic integral has zero mean and finite variance. 
        By Ito isometry, the variance of the hedging error is given by:
        \begin{multline*}
            E\left[|\epsilon(\pi)|^2\right]
            =E\left[\int_0^T\tilde S^2(t-)\left(\pi_t - \frac{\partial \tilde{\mathcal{C}}}{\partial S}(t,S(t-),\alpha(t-))\right)^2\sigma^2(\alpha(t-))dt\right]\\
            +E\left[\int_0^T\int_\mathbb{R}\bigg|\tilde{\mathcal{C}}\left(t,S(t-)e^{\gamma(y,\alpha(t-))},\alpha(t-)+h(\alpha(t-),y)\right) - \tilde{\mathcal{C}}(t,S(t-),\alpha(t-))\right.\\
            \left.- \pi_t\left(e^{\gamma(y,\alpha(t-))}-1\right)\tilde S(t-)\bigg|^2v(dy)dt\right]
        \end{multline*}
        Note that the equation above is a positive process which is a quadratic function of $\pi_t$. 
        The optimal hedge is obtained by minimizing this expression with respect to $\pi_t$.
        We can obtain the first order condition by differentiating the quadratic function.
        \begin{multline*}
            \tilde S^2(t-)\sigma^2(\alpha(t-))\left(\pi_t - \frac{\partial \tilde{\mathcal{C}}}{\partial S}(t,S(t-),\alpha(t-))\right)\\
            +\int_\mathbb{R}(e^{\gamma(y,\alpha(t-))}-1)\tilde S(t-)\left[\left(e^{\gamma(y,\alpha(t-))}-1\right)\tilde S(t-)\pi_t\right.\\
            \left. - \left(\tilde{\mathcal{C}}\left(t,S(t-)e^{\gamma(y,\alpha(t-))},\alpha(t-)+h(\alpha(t-),y)\right) - \tilde{\mathcal{C}}(t,S(t-),\alpha(t-))\right)\right]v(dy)=0.
        \end{multline*}
        Solving yields:
        \begin{equation*}
            \pi(t)= \frac{\sigma^2(\alpha(t-))\frac{\partial \tilde{\mathcal{C}}}{\partial S}(t,S(t-),\alpha(t-)) + \frac{1}{\tilde S(t-)}\int_\mathbb{R} (e^{\gamma(y,\alpha(t-))}-1)(\tilde{\mathcal{C}}_t-\tilde{\mathcal{C}}_{t-})v(dy)}{\sigma^2(\alpha(t))+\int_\mathbb{R}  (e^{\gamma(y,\alpha(t-))}-1)^2v(dy)}
        \end{equation*}
        where $\tilde{\mathcal{C}}_t:=\tilde{\mathcal{C}}(t,S(t-)e^{\gamma(y,\alpha(t-))},\alpha(t-)+h(\alpha(t-),y))$ and $\tilde{\mathcal{C}}_{t-}:=\tilde{\mathcal{C}}(t,S(t-),\alpha(t-)))$.
    \end{proof}

    \section{Fast Fourier Transform Methods}\label{app:FFT}
    We follow the implement representing in \citet{hilpisch2015}, using the trapezoid rule for the integral on formula (\ref{equ:Fourier_call}) , and setting $u_j = \eta(j-1)$, an approximation for $C_0$ is 
    \begin{align*}
        C_0 &\approx Se^{-r_fT} - \frac{\sqrt{SK}e^{-(r_d+r_f)T/2}}{\pi}\sum_{j=1}^N\text{Re}\left[e^{- iu_jk}\varphi_T\left(- u_j-\frac i2\right)\frac{\eta}{u_j^2+\frac14}\right]\\
            &:= Se^{-r_fT} - \frac{\sqrt{SK}e^{-(r_d+r_f)T/2}}{\pi}\text{Re}\left[\sum_{j=1}^Ne^{- iu_jk}\Psi(u_j)\eta\right],
    \end{align*}
    where $k =\ln(\frac{S}{K})$. The effective the integration is now on the interval $[0,N\eta]$. The FFT returns $N$ values of $k$ and we employ a regular spacing of size $\epsilon$, consider the sequence of log-strike $k_m = -b + \epsilon(m-1)$ with $m=1,2,\dots,N$. With this spacing, the FFT algorithm returns $N$ values for log-strike ranging from $-b$ to $b$ with $b = 0.5N\epsilon$. Then
    \begin{equation*}
        C_0 \approx Se^{-r_fT} - \frac{\sqrt{SK}e^{-(r_d+r_f)T/2}}{\pi}\sum_{j=1}^{N}e^{-i\epsilon\eta(j-1)(m-1)}e^{ibu_j}\Psi(u_j)\eta, \quad m = 1,\dots,N.
    \end{equation*}
    Realizing that $\epsilon\eta = \frac{2\pi}{N}$ which implies that a small $\eta$ increase $\epsilon$, introduce weightings according to Simpson's rule such that the call approximation finally takes on the form
    \begin{equation*}
        C_0 \approx Se^{-r_fT} - \frac{\sqrt{SK}e^{-(r_d+r_f)T/2}}{\pi}\sum_{j=1}^{N}e^{-i\epsilon\eta(j-1)(m-1)}e^{ibu_j}\Psi(u_j)\frac{\eta}{3}(3+(-1)^j-\diamond_{j-1})
    \end{equation*}
    for $m = 1,\dots,N$ and where $\diamond_n$ is the Kronecker delta function which takes value one for $n=0$ and zero otherwise.

\end{appendix}
\end{footnotesize}
\newpage
\bibliographystyle{abbrvnat}
\bibliography{biblio}
\end{document}